\documentclass[journal,twocolumn]{IEEEtran}
\usepackage{amsmath}
\usepackage{lipsum}
\usepackage{graphicx}
\usepackage[utf8]{inputenc}
\usepackage{siunitx}
\usepackage{upgreek}
\usepackage{stackengine}
\usepackage{algorithmic}
\usepackage{algorithm,algorithmic} 
\usepackage{amssymb}
\usepackage{optidef}
\usepackage{xcolor}
\usepackage{enumitem}
\usepackage[english]{babel}
\usepackage{amsthm}
\usepackage{amssymb}
\newtheorem{theorem}{Theorem}[section]
\newtheorem{lemma}[theorem]{Lemma}
\usepackage{authblk}
\usepackage{cuted}
\newlist{steps}{enumerate}{1}
\setlist[steps, 1]{label = \textbf{Step} \arabic*:}
\usepackage{cite}
\usepackage[nolist]{acronym}
\usepackage{tabularx}
\usepackage{pgfplotstable}
\usepgfplotslibrary{patchplots}
\usepackage{authblk}
\usepackage{dblfloatfix}    
\usepackage{caption}
\usepackage{subcaption} 
\usepackage[capitalise]{cleveref}
\Crefname{equation}{Eq.\!}{Eqs.\!}
\Crefname{figure}{Fig.\!}{Figs.\!}
\Crefname{tabular}{Tab.\!}{Tabs.\!}
\Crefname{section}{Section\!}{Sections.\!}
\Crefname{appendices}{Appendix\!}{Appendix.\!}

\usepackage{booktabs}

\begin{document}
\def\GHz{\text{GHz}} 
\def\THz{\text{THz}} 
\def\Nled{N}  
\def\scheme{\text{proposed}} 
\def\Pmatrix{\mathbf{C}}  
\def\Rfec{\mathrm{R}_{\text{FEC}}}
\def\Rbf{\mathrm{R}_{\text{bf}}}
\def\Rfecone{\mathrm{R}_{\text{FEC}_{1}}}
\def\Rfectwo{\mathrm{R}_{\text{FEC}_{2}}}
\def\Pbcf{c}
\def\DeltaIntf{\overline{\mathbf{\Delta}}_{j+1}}
\def\PdIntrf{\overline{\mathbf{p}}_{j+1}}
\def\Roimac{\mathrm{R}_{j}\!\left(\Delta_{j},\mathbf{p}_{j},\AIntrf\right)}
\def\Ru{\mathrm{R}_{j}\!\left(\Delta_{j},\mathbf{p}^{u}_{j},\AIntrf\right)}
\def\Ruopt{\mathrm{R}_{j}\!\left(\Delta_{j},\mathbf{p}^{u}_{j},\AIntrf^{\star}\right)}
\def\Rsdt{\mathrm{R}^{\mathrm{SDT}}_{j}\!\left(\Delta_{j},\mathbf{p}_{j},\AIntrf\right)}
\def\Rsdtopt{\mathrm{R}^{\mathrm{SDT}}_{j}\!\left(\Delta_{j},\mathbf{p}_{j},\AIntrf^{\star}\right)}
\def\RsdtFinal{\mathrm{R}^{\mathrm{SDT}}_{j}\!\left(\Delta^{\star}_{j},\mathbf{p}^{\star}_{j},\AIntrf^{\star}\right)}
\def\RsdtoptOne{\mathrm{R}^{\mathrm{SDT}}_{1}}
\def\RsdtoptTwo{\mathrm{R}^{\mathrm{SDT}}_{2}}
\def\CapNOMA{\mathrm{C}_{j}}

\def\AIntrf{\mathbf{a}_{\widehat{y}_{j}}}
\def\Phip{\Phi\left(\mathbf{p}_{j},\hat{\mathbf{p}}_{j}\right)}
\def\Wiphat{W\!\left(i,\Delta_{j},\hat{\mathbf{p}}_{j},\AIntrf^{\star}\right)}
\def\Wiprimephat{W\!\left(i^{\prime},\Delta_{j},\hat{\mathbf{p}}_{j},\AIntrf^{\star}\right)}
\def\Wip{W\!\!\left(i,\Delta_{j},\mathbf{p}_{j},\AIntrf^{\star}\right)}

\def\CondOnePower{c^{\mathrm{p}}}
\def\CondOneRate{c^{\mathrm{r}}}

\begin{acronym}

\acro{5G-NR}{5G New Radio}
\acro{3GPP}{3rd Generation Partnership Project}
\acro{AC}{address coding}
\acro{ACF}{autocorrelation function}
\acro{ACR}{autocorrelation receiver}
\acro{ADC}{analog-to-digital converter}
\acrodef{aic}[AIC]{Analog-to-Information Converter}     
\acro{AIC}[AIC]{Akaike information criterion}
\acro{AIR}{achievable information rate}
\acro{aric}[ARIC]{asymmetric restricted isometry constant}
\acro{arip}[ARIP]{asymmetric restricted isometry property}

\acro{AR}{achievable rate}
\acro{ASR}{achievable sum rate}
\acro{AUB}{asymptotic union bound}
\acrodef{awgn}[AWGN]{Additive White Gaussian Noise}     
\acro{AWGN}{additive white Gaussian noise}

\acro{APSK}[PSK]{asymmetric PSK} 

\acro{waric}[AWRICs]{asymmetric weak restricted isometry constants}
\acro{warip}[AWRIP]{asymmetric weak restricted isometry property}

\acro{BAA}{Blahut-Arimoto algorithm}
\acro{BCH}{Bose, Chaudhuri, and Hocquenghem}        
\acro{BCHC}[BCHSC]{BCH based source coding}
\acro{BEP}{bit error probability}
\acro{BFC}{block fading channel}
\acro{BG}[BG]{Bernoulli-Gaussian}
\acro{BGG}{Bernoulli-Generalized Gaussian}
\acro{BMD}{bit metric decoder}
\acro{BPAM}{binary pulse amplitude modulation}
\acro{BPDN}{Basis Pursuit Denoising}
\acro{bpcu}{bits per channel use}
\acro{BPPM}{binary pulse position modulation}
\acro{BPSK}{binary phase shift keying}
\acro{BPZF}{bandpass zonal filter}
\acro{BRYC}{blue, red, yellow, and cyan}
\acro{BSC}{binary symmetric channels}              
\acro{BU}[BU]{Bernoulli-uniform}
\acro{BER}{bit error rate}
\acro{BS}{base station}

\acro{CP}{Cyclic Prefix}
\acrodef{cdf}[CDF]{cumulative distribution function}   
\acro{CDF}{cumulative distribution function}
\acrodef{c.d.f.}[CDF]{cumulative distribution function}
\acro{CCDF}{complementary cumulative distribution function}
\acrodef{ccdf}[CCDF]{complementary CDF}               
\acrodef{c.c.d.f.}[CCDF]{complementary cumulative distribution function}
\acro{CCDM}{constant composition distribution matcher}
\acro{CD}{cooperative diversity}

\acro{CDMA}{Code Division Multiple Access}
\acro{ch.f.}{characteristic function}
\acro{CIR}{channel impulse response}
\acro{CMIMO}{color-space multiple input and multiple output}
\acro{cosamp}[CoSaMP]{compressive sampling matching pursuit}
\acro{CR}{cognitive radio}
\acro{cs}[CS]{compressed sensing}                   
\acrodef{cscapital}[CS]{Compressed sensing} 
\acrodef{CS}[CS]{compressed sensing}
\acro{CSI}{channel state information}
\acro{CSK}{color-shift keying}
\acro{CCSDS}{consultative committee for space data systems}
\acro{CC}{convolutional coding}
\acro{Covid19}[COVID-19]{Coronavirus disease}

\acro{DAA}{detect and avoid}
\acro{DAB}{digital audio broadcasting}
\acro{dB}{decibel}
\acro{DC}{direct current}
\acro{DCT}{discrete cosine transform}
\acro{dft}[DFT]{discrete Fourier transform}
\acro{DM}{distribution matching}
\acro{DMT/OFDM}{discrete multitone or orthogonal frequency-division multiplexing}
\acro{DoF}{degree of freedom}
\acro{DR}{distortion-rate}
\acro{DS}{direct sequence}
\acro{DS-SS}{direct-sequence spread-spectrum}
\acro{DTR}{differential transmitted-reference}
\acro{DVB-H}{digital video broadcasting\,--\,handheld}
\acro{DVB-T}{digital video broadcasting\,--\,terrestrial}
\acro{DL}{downlink}
\acro{DSSS}{Direct Sequence Spread Spectrum}
\acro{DFT-s-OFDM}{Discrete Fourier Transform-spread-Orthogonal Frequency Division Multiplexing}
\acro{DAS}{distributed antenna system}
\acro{DNA}{Deoxyribonucleic Acid}

\acro{EC}{European Commission}
\acro{EE}{energy efficiency}
\acro{EED}[EED]{exact eigenvalues distribution}
\acro{EIRP}{Equivalent Isotropically Radiated Power}
\acro{ELP}{equivalent low-pass}
\acro{eMBB}{Enhanced Mobile Broadband}
\acro{EMF}{electric and magnetic fields}
\acro{EU}{European union}

\acro{FC}[FC]{fusion center}
\acro{FCC}{Federal Communications Commission}
\acro{FEC}{forward error correction}
\acro{FER}{frame error rate}
\acro{FFT}{fast Fourier transform}
\acro{FH}{frequency-hopping}
\acro{FH-SS}{frequency-hopping spread-spectrum}
\acrodef{FS}{Frame synchronization}
\acro{FSK}{frequency shift keying}
\acro{FSO}{free space optical}
\acro{FSsmall}[FS]{frame synchronization}  
\acro{FDMA}{Frequency Division Multiple Access}

\acro{GA}{Gaussian approximation}
\acro{Gbps}{giga bits per second}
\acro{GF}{Galois field }
\acro{GG}{Generalized-Gaussian}
\acro{GLIM}{generalized LED index modulation }
\acro{GIC}[GIC]{generalized information criterion}
\acro{GLRT}{generalized likelihood ratio test}
\acro{GMI}{generalized mutual information}
\acro{GPS}{Global Positioning System}
\acro{GMSK}{Gaussian minimum shift keying}
\acro{GS}{geometric shaping}
\acro{GSMA}{Global System for Mobile communications Association}

\acro{HAP}{high altitude platform}

\acro{IDR}{information distortion-rate}
\acro{IFFT}{inverse fast Fourier transform}
\acro{iht}[IHT]{iterative hard thresholding}
\acro{i.i.d.}{independent, identically distributed}
\acro{IM/DD}{intensity modulation and direct detection}
\acro{IoT}{Internet of Things}                      
\acro{IR}{impulse radio}
\acro{lric}[LRIC]{lower restricted isometry constant}
\acro{lrict}[LRICt]{lower restricted isometry constant threshold}
\acro{ISI}{intersymbol interference}
\acro{ITU}{International Telecommunication Union}
\acro{ICNIRP}{International Commission on Non-Ionizing Radiation Protection}
\acro{IEEE}{Institute of Electrical and Electronics Engineers}
\acro{ICES}{IEEE international committee on electromagnetic safety}
\acro{ICI}{inter-channel interference}
\acro{IEC}{International Electrotechnical Commission}
\acro{IARC}{International Agency on Research on Cancer}
\acro{IS-95}{Interim Standard 95}

\acro{KKT}{Karush-Kuhn-Tucker}

\acro{LDPC} {low-density parity-check}
\acro{LED} {light emitting diode}
\acro{LEO}{low earth orbit}
\acro{LF}{likelihood function}
\acro{LLF}{log-likelihood function}
\acro{LLR}{log-likelihood ratio}
\acro{LLRT}{log-likelihood ratio test}
\acro{LOS}{Line-of-Sight}
\acro{LRT}{likelihood ratio test}
\acro{wlric}[LWRIC]{lower weak restricted isometry constant}
\acro{wlrict}[LWRICt]{LWRIC threshold}
\acro{LPWAN}{low power wide area network}
\acro{LoRaWAN}{Low power long Range Wide Area Network}
\acro{NLOS}{non-line-of-sight}

\acro{MAC}{multiple access channel}
\acro{MAP}{maximum a posteriori probability}
\acro{MB}{multiband}
\acro{MC}{multicarrier}
\acro{MDS}{mixed distributed source}
\acro{MF}{matched filter}
\acro{m.g.f.}{moment generating function}
\acro{MI}{mutual information}
\acro{MIMO}{multiple-input multiple-output}
\acro{MISO}{multiple-input single-output}
\acrodef{maxs}[MJSO]{maximum joint support cardinality}                       
\acro{ML}[ML]{maximum likelihood}
\acro{MMSE}{minimum mean-square error}
\acro{MMV}{multiple measurement vectors}
\acrodef{MOS}{model order selection}
\acro{M-PSK}[${M}$-PSK]{$M$-ary phase shift keying}                       
\acro{M-APSK}[${M}$-PSK]{$M$-ary asymmetric PSK} 

\acro{M-QAM}[$M$-QAM]{$M$-ary quadrature amplitude modulation}
\acro{MRC}{maximal ratio combiner}                  
\acro{maxs}[MSO]{maximum sparsity order}                                      
\acro{M2M}{machine to machine}                                                
\acro{MUI}{multi-user interference}
\acro{mMTC}{massive Machine Type Communications}      
\acro{mm-Wave}{millimeter-wave}
\acro{MP}{mobile phone}
\acro{MPE}{maximum permissible exposure}

\acro{NB}{narrowband}
\acro{NBI}{narrowband interference}
\acro{NHS}{National Health Service}
\acro{NLA}{nonlinear sparse approximation}
\acro{NLOS}{Non-Line of Sight}
\acro{NOMA}{nonorthogonal multiple access}
\acro{NTIA}{National Telecommunications and Information Administration}
\acro{NTP}{National Toxicology Program}

\acro{OC}{optimum combining}                             
\acro{OC}{optimum combining}
\acro{ODE}{operational distortion-energy}
\acro{ODR}{operational distortion-rate}
\acro{OFDM}{orthogonal frequency-division multiplexing}
\acro{OMAC}{optical multiple access channel}
\acro{OMA}{orthogonal multiple access}
\acro{omp}[OMP]{orthogonal matching pursuit}
\acro{OSMP}[OSMP]{orthogonal subspace matching pursuit}
\acro{OQAM}{offset quadrature amplitude modulation}
\acro{OQPSK}{offset QPSK}
\acro{OFDMA}{Orthogonal Frequency-division Multiple Access}
\acro{OPEX}{Operating Expenditures}
\acro{OSINR}[\textrm{OSINR}]{optical signal-to-interference plus noise ratio}
\acro{OSNR}[\textrm{OSNR}]{optical signal-to-noise ratio} 
\acro{OQPSK/PM}{OQPSK with phase modulation}
\acro{OWC}{Optical wireless communications}

\acro{PAM}{pulse amplitude modulation}
\acro{PAR}{peak-to-average ratio}
\acro{PAPR}{peak-to-average power ratio }
\acro{PAS}{probabilistic amplitude shaping}
\acro{PCM}{pairwise coded modulation}
\acrodef{pdf}[PDF]{probability density function}                      
\acro{PDF}{probability density function}
\acrodef{p.d.f.}[PDF]{probability distribution function}
\acro{PDP}{power dispersion profile}
\acro{PD-NOMA}{power domain NOMA}
\acro{PMF}{probability mass function}                             
\acrodef{p.m.f.}[PMF]{probability mass function}
\acro{PN}{pseudo-noise}
\acro{PPM}{pulse position modulation}
\acro{PRake}{Partial Rake}
\acro{PS}{probabilistic shaping}
\acro{PSD}{power spectral density}
\acro{PSEP}{pairwise synchronization error probability}
\acro{PSK}{phase shift keying}
\acro{PD}{photo-detector}
\acro{P2P}{point-to-point}
\acro{8-PSK}[$8$-PSK]{$8$-phase shift keying}

\acro{QAM}{Quadrature Amplitude Modulation}
\acro{QLED}{quadrichromatic LED}
\acro{QPSK}{quadrature phase shift keying}
\acro{OQPSK/PM}{OQPSK with phase modulator}

\acro{RD}[RD]{raw data}
\acro{RDL}{"random data limit"}
\acro{RF}{radio frequency}
\acro{ric}[RIC]{restricted isometry constant}
\acro{rict}[RICt]{restricted isometry constant threshold}
\acro{rip}[RIP]{restricted isometry property}
\acro{RGB}{red/green/blue}
\acro{RGBA}{red/green/blue/amber}
\acro{ROC}{receiver operating characteristic}
\acro{rq}[RQ]{Raleigh quotient}
\acro{RS}[RS]{Reed-Solomon}
\acro{RSC}[RSSC]{RS based source coding}
\acro{r.v.}{random variable}                               
\acro{R.V.}{random vector}
\acro{RMS}{root mean square}
\acro{RFR}{radiofrequency radiation}
\acro{RIS}{reconfigurable intelligent surface}
\acro{RNA}{RiboNucleic Acid}

\acro{SA}[SA-Music]{subspace-augmented MUSIC with OSMP}
\acro{SCBSES}[SCBSES]{Source Compression Based Syndrome Encoding Scheme}
\acro{SCM}{sample covariance matrix}
\acro{SDT}{sparse-dense transmission}
\acro{SE}{spectral efficiency}
\acro{SEP}{symbol error probability}
\acro{SG}[SG]{sparse-land Gaussian model}
\acro{SIC}{successive interference cancellation}
\acro{SIMO}{single-input multiple-output}
\acro{SINR}{signal-to-interference plus noise ratio}
\acro{SIR}{signal-to-interference ratio}
\acro{SISO}{single-input single-output}
\acro{SMD}{symbol metric decoder}
\acro{SMV}{single measurement vector}
\acro{SMT}{space modulation technique}
\acro{SM}{spatial modulation}
\acro{SNR}[\textrm{SNR}]{signal-to-noise ratio} 
\acro{sp}[SP]{subspace pursuit}
\acro{SPD}{spectral power distribution}
\acro{SS}{spread spectrum}
\acro{SW}{sync word}
\acro{SAR}{specific absorption rate}
\acro{SSB}{synchronization signal block}
\acro{sss}[SpaSoSEnc]{sparse source syndrome encoding}

\acro{TH}{time-hopping}
\acro{TLED}{TriLED}
\acro{ToA}{time-of-arrival}
\acro{TR}{transmission rate}
\acro{TW}{Tracy-Widom}
\acro{TWDT}{TW Distribution Tail}
\acro{TCM}{trellis coded modulation}
\acro{TDD}{time-division duplexing}
\acro{TDMA}{Time Division Multiple Access}

\acro{UAV}{unmanned aerial vehicle}
\acro{UE}{user equipment}
\acro{UL}{uplink}
\acro{uOFDM}{unipolar orthogonal frequency-division multiplexing}
\acro{uric}[URIC]{upper restricted isometry constant}
\acro{urict}[URICt]{upper restricted isometry constant threshold}
\acro{URLLC}{Ultra Reliable Low Latency Communications}
\acro{UWB}{ultrawide band}
\acro{UWBcap}[UWB]{Ultrawide band}

\acro{WDM}{wavelength division multiplexing}
\acro{WHO}{World Health Organization}
\acro{WiM}[WiM]{weigh-in-motion}
\acro{Wi-Fi}{wireless fidelity}
\acro{WLAN}{wireless local area network}
\acro{wm}[WM]{Wishart matrix}                               
\acroplural{wm}[WM]{Wishart matrices}
\acro{WMAN}{wireless metropolitan area network}
\acro{WPAN}{wireless personal area network}
\acro{wric}[WRIC]{weak restricted isometry constant}
\acro{wrict}[WRICt]{weak restricted isometry constant thresholds}
\acro{wrip}[WRIP]{weak restricted isometry property}
\acro{WSN}{wireless sensor network}                        
\acro{WSS}{wide-sense stationary}
\acro{wuric}[UWRIC]{upper weak restricted isometry constant}
\acro{wurict}[UWRICt]{UWRIC threshold}

\acro{VLC}{visible light communication}
\acro{VL}{visible light}
\acro{VPN}{virtual private network} 
\acro{FSO}{free space optics}
\acro{IoST}{Internet of space things}

\acro{GSM}{Global System for Mobile Communications}
\acro{2G}{second-generation cellular network}
\acro{3G}{third-generation cellular network}
\acro{4G}{fourth-generation cellular network}
\acro{5G}{5th-generation cellular network}	
\acro{6G}{6th-generation}	
\acro{gNB}{next generation node B base station}
\acro{NR}{New Radio}
\acro{UMTS}{Universal Mobile Telecommunications Service}
\acro{LTE}{Long Term Evolution}

\acro{QoS}{Quality of Service}
\end{acronym}
	\pgfplotsset{every axis/.append style={
			line width=1pt,
			legend style={font=\large, at={(0.97,0.85)}}},
	} %
	\pgfplotsset{compat=1.18, legend style={fill=white, fill opacity=0.85, draw opacity=0.9, text opacity=1}}
\title{Probabilistic Constellation Shaping for Enhancing Spectral Efficiency in NOMA VLC Systems}
\author{Amanat Kafizov,~\IEEEmembership{Student Member,~IEEE}, Ahmed Elzanaty,~\IEEEmembership{Senior Member,~IEEE}, Mohamed-Slim Alouini,~\IEEEmembership{Fellow,~IEEE}
\thanks{A. Kafizov and  M.-S. Alouini are with Computer, Electrical, and Mathematical Science and Engineering (CEMSE) Division, King Abdullah University of Science and Technology (KAUST), Thuwal, Saudi Arabia (email: \{amanat.kafizov, slim.alouini\}@kaust.edu.sa).}
\thanks{A. Elzanaty is with the Institute for Communication Systems (ICS), Home of
the 5G and 6G Innovation Centres (5GIC and 6GIC), University of Surrey,
Guildford GU2 7XH, United Kingdom (e-mail: a.elzanaty@surrey.ac.uk).}
}
\maketitle
\begin{abstract}
The limited modulation bandwidth of the \acp{LED} presents a challenge in the development of practical high-data-rate \ac{VLC} systems. In this paper, a novel adaptive coded \ac{PS}-based \ac{NOMA} scheme is proposed to improve \ac{SE} of \ac{VLC} systems in multiuser uplink communication scenarios. The proposed scheme adapts its rate to the \ac{OSNR} by utilizing non-uniformly distributed discrete constellation symbols and low complexity channel encoder. Furthermore, an alternate optimization algorithm is proposed to determine the optimal channel coding rate, constellation spacing, and \ac{PMF} of each user. The extensive numerical results show that the proposed \ac{PS}-based \ac{NOMA} scheme closely approaches the capacity of \ac{NOMA} with fine granularity. Presented results demonstrate the effectiveness of our scheme in improving the \ac{SE} of \ac{VLC} systems in multiuser scenarios. For instance, our scheme exhibits substantial \ac{SE} gains over existing schemes, namely, the \ac{PCM}, \ac{GS}, and uniform-distribution schemes. These findings highlight the potential of our approach to significantly enhance \ac{VLC} systems.

\end{abstract}
	\acresetall 
\section{Introduction}\label{sec:Inro}
The exponential growth in the number of \ac{IoT} devices has led to a significant demand for high data rates in wireless traffic. However, traditional \ac{RF} communication systems face challenges in meeting this demand due to limited spectrum resources. As a promising high-speed alternative and complementary technique to \ac{RF}, \ac{VLC}, which utilizes \ac{LED} devices for both illumination and data transmission, has attracted considerable research attention \cite{khalighi2014survey}. Nonetheless, the development of practical high-data-rate \ac{VLC} systems is hindered by the narrow modulation bandwidth of \acp{LED} \cite{marshoud2018optical}. As a result, numerous techniques, encompassing advanced multiple access schemes, \ac{MIMO}, \ac{RIS}, and frequency reuse methods, have been explored to enhance \ac{SE} of \ac{VLC} systems. Nevertheless, a shared limitation in these methodologies is their reliance on a uniform symbol distribution, resulting in a gap loss from the channel capacity \cite{mathur2021survey,forney1984efficient,wang2023transmit,abdelhady2020visible}.

 Recently, constellation shaping techniques have emerged as a method to improve the \ac{SE} of communication systems, mitigating the aforementioned gap loss. In general, to achieve the channel capacity with a given transmitted power constraint, the input to the channel should have a capacity-achieving distribution. However, the actual input distribution often deviates from the capacity-achieving one, resulting in a shaping gap between the \ac{AR} and the channel capacity \cite{kschischang1993optimal}. \Ac{PS} and \ac{GS} are widely known techniques that aim to reduce this shaping gap by optimizing the probability and location of the constellation symbols, respectively. Compared to \ac{GS}, \ac{PS} offers greater flexibility for granularity and enables easy implementation of Gray mapping \cite{cho2019probabilistic,javed2021probabilistic}. Importantly, it's worth noting that constellation shaping doesn't conflict with existing techniques used to enhance the \ac{SE} of \ac{VLC} systems. Instead, it can be combined with these techniques to further enhance the \ac{SE} of \ac{VLC} systems.

Authors in \cite{bocherer2015bandwidth} propose a practical \ac{PS} technique based on reverse concatenation architecture, where distribution matcher is applied before a systematic channel encoder at the transmitter. In \cite{bocherer2015bandwidth}, uniform information bits are practicably transformed into symbols with desired distribution by \ac{CCDM}. The symbols with low energy are transmitted more frequently than the high energy symbols. Thus, the available bandwidth is efficiently utilized. The advantages of this architecture have led to its widespread adoption in optical communications. For example, the work in \cite{elzanaty2020adaptive} proposes an adaptive modulation scheme based on \ac{PS} for \ac{FSO} channels, providing a solution for \ac{FSO} backhauling in
terrestrial and satellite communication systems to achieve higher \ac{SE}. In \cite{kafizov2022probabilistic}, authors proposes adaptive \ac{SM} scheme with \ac{PS}, where either the spatial or constellation symbols are probabilistically shaped,  to improve the \ac{TR} of \ac{VLC} systems. However, existing works primarily focus on \ac{P2P} communication scenarios.


In real-world \ac{VLC} systems, where a single \ac{LED} transmitter is typically expected to accommodate multiple users, addressing the challenge of serving multiple users has led researchers to propose advanced schemes for \ac{MAC}. While \ac{OMA} techniques have been designed for \ac{VLC} systems, researchers are exploring \ac{NOMA} techniques to efficiently increase the \ac{SE} of \ac{VLC} systems \cite{feng2019joint}. Among the various \ac{NOMA} techniques, power domain \ac{NOMA} stands out as a widely adopted approach, wherein multiple users are served over the available time and frequency resources using different power levels \cite{sadat2022survey}. \Ac{NOMA} offers several benefits such as low transmission latency, improved \ac{SE} and higher
cell-edge throughput. The \ac{NOMA} symbol is generated by superposing signals from multiple users, while \ac{SIC} is employed at the receiver for signal detection \cite{yin2016performance}. 

In order to examine the interaction between \ac{NOMA} and \ac{VLC}, researchers in \cite{yin2016performance} presented a mathematical framework to assess the average performance of \ac{NOMA} \ac{VLC} systems. The analytical findings demonstrated that \ac{NOMA} \ac{VLC} surpasses \ac{OMA} \ac{VLC} in terms of the coverage probability, particularly in high \ac{OSNR} conditions. The study also revealed that by increasing the discrepancy in channel gains between paired users and selecting appropriate \acp{LED}, the benefits of \ac{NOMA} can be further amplified. In \cite{feng2019joint}, the authors aim to maximize the sum rate by optimizing the power allocation for each user in a power-line-fed \ac{VLC} network. The efficient channel estimation method for minimizing the \ac{BER} in the optical \ac{MIMO} \ac{NOMA}-\ac{VLC} is proposed in \cite{lin2017experimental}. In \cite{xiao2019hybrid} and \cite{obeed2020user}, the \ac{NOMA}-based hybrid \ac{VLC}/\ac{RF} systems are presented with design objectives as outage probability and the sum rate, respectively. However, these works consider uniform  signaling, which leads to the shaping gap. 

In contrast, authors in \cite{xu2020investigation} examine the application of \ac{PS} in \ac{NOMA} \ac{VLC}. However, this study encounters significant issues. Firstly, the adopted analysis methods, specifically the choice of \ac{AR}, are unsuitable for the \ac{VLC} domain. Secondly, the sum rate constraint of the capacity region of \ac{OMAC} is disregarded. Thirdly, essential properties such as limited average and/or peak optical power and the non-negativity of optical intensity are violated. Consequently, there is a lack of extensive research in analyzing the performance of \ac{PS} in \ac{NOMA} \ac{VLC} systems, which necessitates further exploration.

\begin{table}[b]
  \centering
  \caption{Comparison between proposed and typical works.}
  \begin{tabular}{lcccc}
    \toprule
    \textbf{Work} & \textbf{Constellation} & \textbf{Frequency} & \textbf{Channel} &  \textbf{Optimization}\\
    \textbf{} & \textbf{Shaping} & \textbf{Band Type} & \textbf{Type} &  \textbf{Objective} \\
    \midrule
    Proposed & Yes & \ac{VLC} & \ac{MAC} & Transmission \\
     &  &  &  &  rate\\
     \hline
    \cite{bocherer2015bandwidth} & Yes & \Ac{RF} & \ac{P2P} & Transmission\\
     &  &  &  &  rate\\
     \hline
    \cite{elzanaty2020adaptive} & Yes & \Ac{FSO} &  \ac{P2P} & Achievable\\
      &  &  &  &  rate\\
     \hline
    \cite{kafizov2022probabilistic} & Yes & \Ac{VLC} & \ac{P2P} & Transmission \\
      &  &  &  &  rate\\
     \hline
     \cite{feng2019joint} & No & \Ac{VLC} & \ac{MAC} & Sum rate\\
     \hline
     \cite{yin2016performance} & No & \Ac{VLC} & \ac{MAC} &  Coverage\\
      &  &  &  &  probability\\
     \hline
     \cite{lin2017experimental} & No & \Ac{VLC} & \ac{MAC} & \Ac{BER}\\
      \hline
     \cite{xiao2019hybrid} & No & Hybrid & \ac{MAC} & Outage \\
      &  & \Ac{VLC}/\ac{RF} &  & Probability\\
     \hline
     \cite{obeed2020user} & No & Hybrid & \ac{MAC} & Sum rate\\
     &  & \Ac{VLC}/\ac{RF} &  & \\
      \hline
     \cite{xu2020investigation} & No & Unsuitable & \ac{MAC} & Entropy\\
       &  & for \ac{VLC} &  & \\
    \bottomrule
  \end{tabular}
  \label{table:ch4:ComparisonTypicalWorks}
\end{table}

In \cref{table:ch4:ComparisonTypicalWorks}, we provide a succinct overview of the key attributes of the discussed works and our proposed scheme for better clarity and comparison. \cref{table:ch4:ComparisonTypicalWorks} summarizes key characteristics of these works, focusing on aspects such as constellation shaping, frequency band type, channel type and optimization objective. As highlighted in the table, traditional \ac{NOMA} \ac{VLC} systems typically adopt uniform input signals, where constellation shaping is an area yet to be fully explored to mitigate shaping loss for each user. Unlike existing schemes, our proposed adaptive coded \ac{PS}-based \ac{NOMA} scheme effectively adapts its rate to the \ac{OSNR} by utilizing non-uniformly distributed discrete constellation symbols. This adaptive strategy, coupled with a low-complexity channel encoder and an alternate optimization algorithm, leads to substantial gains in \ac{SE} for \ac{VLC} systems in multiuser uplink scenarios. The extensive numerical results showcase the scheme's ability to closely approach the capacity of \ac{NOMA}, outperforming existing schemes such as \ac{PCM}, \ac{GS}, and uniform-distribution schemes.

\begin{figure}[t]
	\centering	\includegraphics[width=0.7\linewidth]{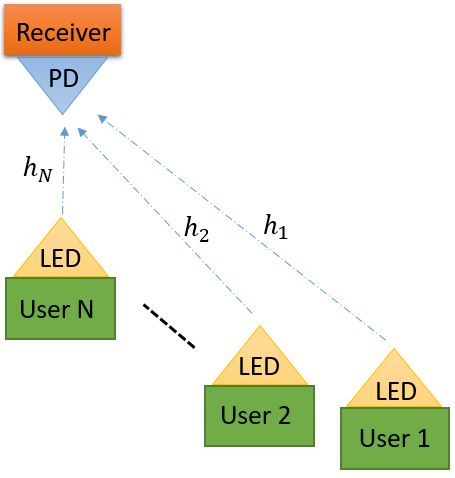}
	\caption{Uplink \ac{NOMA} transmission scenario for $\Nled$ users.}
	\label{fig:ch4:NOMAscenario}
\end{figure}
\begin{figure*}
	\centering
	\includegraphics[width=0.992\linewidth]{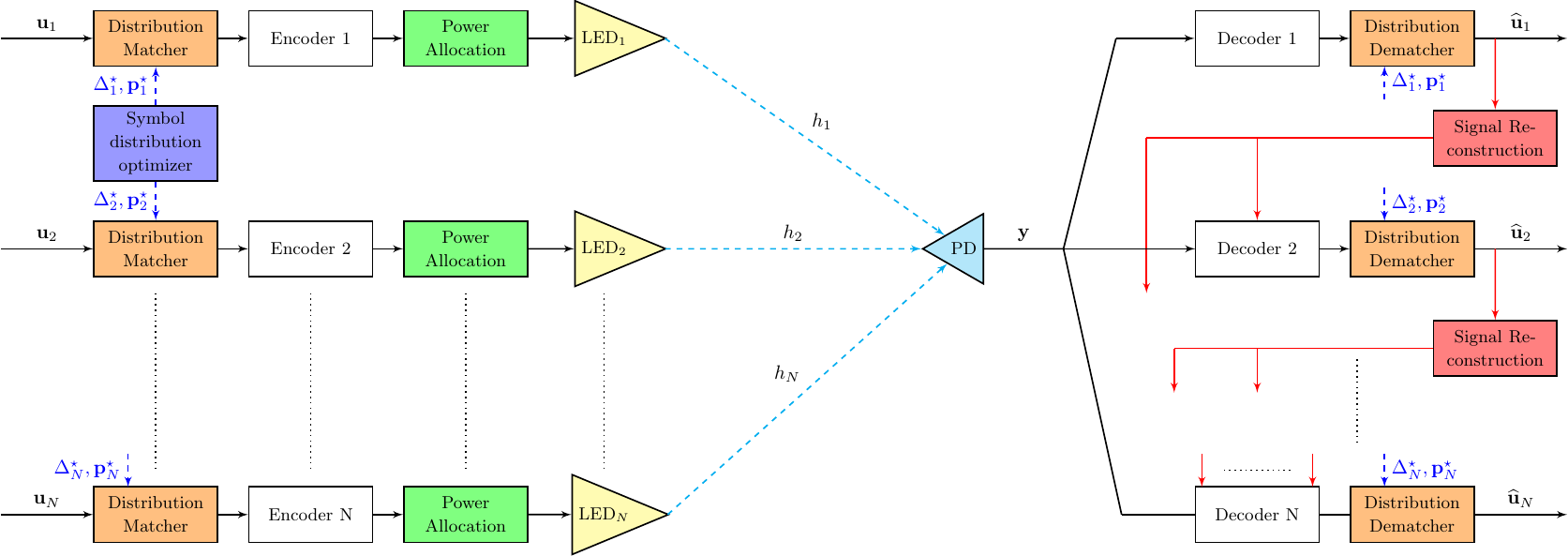}
	\caption{Architecture of \ac{NOMA} with \ac{PS} in \ac{VLC} systems.}
	\label{fig:ch4:NOMAarchitectue}
\end{figure*}
\subsection{Main Contributions}
In this paper, we propose a novel adaptive coded \ac{PS}-based \ac{NOMA} scheme designed to further improve the \ac{SE} of \ac{VLC} systems in multiuser uplink scenarios.  At the encoder of the $\scheme$ scheme, the uniform information bits are practicably transformed into unipolar $M$-PAM symbols with desired distribution by \ac{CCDM} \cite{schulte2015constant}, laying the foundation for adaptability in our proposed scheme. An efficient \ac{FEC} encoder is then applied to generate uniformly distributed parity bits. The uniformly distributed unipolar $M$-\ac{PAM} symbols are produced from the parity bits using the binary demapper. After that, multiplexer appends probabilistically shaped and uniformly distributed symbols to form a codeword. At the decoder, distribution dematching is performed after \ac{FEC} decoding. There is an optimization problem for each user that needs to be solved to determine the capacity-achieving distribution. By probabilistically shaping the distribution of the unipolar $M$-PAM symbols and optimizing their constellation spacing, the \ac{TR} of our proposed scheme approaches the capacity of \ac{NOMA} with fine granularity. The main contributions are summarized as follows:
\begin{itemize}
    \item We propose a novel adaptive coded \ac{PS}-based \ac{NOMA} \ac{VLC} scheme, where the spacing between unipolar $M$-PAM symbols and their \ac{PMF} are optimized for each user depending on the \ac{OSNR}.
    \item In contrast to \ac{P2P} communication, our work incorporates full consideration of multiuser interference to establish the theoretical foundation for derivation of the capacity of uplink \ac{NOMA} \ac{VLC}.
     \item We derive both the \ac{AR} and \ac{TR}
of the proposed scheme, providing deeper insights into
its performance characteristics. 
    \item A multi-variable optimization problem, formulated with the \ac{TR} objective function, is designed for each user to obtain the optimal distribution of constellation symbols. These optimization problems need to be solved sequentially, following an order that is inverse to the order of \ac{SIC}. 
    \item Solving optimization problem in original form is intrinsically complex due to the non-convex nature of the original problem. To simplify the problem, we propose an iterative approach where one variable is optimized while the other is fixed alternately.
    
    \item The optimization subproblem of alternating iterative approach is still non-convex. To address the non-convexity of optimization subproblem, we introduce a surrogate function, as proposed in \cite{lange2000optimization} and \cite{vontobel2008generalization}, to convexify the subproblem. 
    
    \item For convexified optimization subproblem, we use \ac{KKT} conditions to find the optimal input distribution that approaches the capacity of \ac{NOMA}.
    \item We provide an algorithm to compute the capacity-achieving distribution for the proposed scheme.
    \item The transmission and channel coding rates are adapted to \acp{OSNR}. 
    \item We assess the performance of the proposed scheme in terms of \ac{SE} and \ac{FER} under various conditions of \ac{OSNR}. Additionally, we compare our results with the lower and upper capacity bounds from \cite{chaaban2017capacity}, as well as with those obtained using \ac{PCM}, \ac{GS}, and uniform-distribution schemes. 
\end{itemize}

\subsection{Paper Organization and Notations}
The rest of the paper is organized as follows. \cref{sec:ch4:SignalModel} describes the \ac{NOMA} \ac{VLC} system model. The proposed scheme is presented in \cref{sec:ch4:ProposedScheme}. In \cref{sec:ch4:CapAnalysis}, we provide the rate analysis for each user. \cref{sec:ch4:RateAdaptation} describes optimization problem and proposed algorithm. The performance in terms of the \ac{SE} and \ac{FER} are shown in \cref{sec:ch4:PerfAnalysis}. Finally, \cref{sec:ch4:Conclusion} concludes our work.

This paper uses the following notations. Row vectors and matrices are represented by the lowercase and uppercase bold letters, respectively. The function $(\cdot)^{T}$ denotes the transpose. The notation $\mathbb{P}\{\cdot\}$ is the probability of an event, while $\mathbb{E}\left(\cdot\right)$ is the expectation of a \ac{r.v.}. The logarithm function denoted as $\log(\cdot)$ has base two and acts element-wise on vectors. The optimal solution is denoted by $\left(\cdot\right)^\star$.

\section{Systems Model}\label{sec:ch4:SignalModel}
\begin{figure*}
	\centering
	\includegraphics[width=0.99\textwidth]{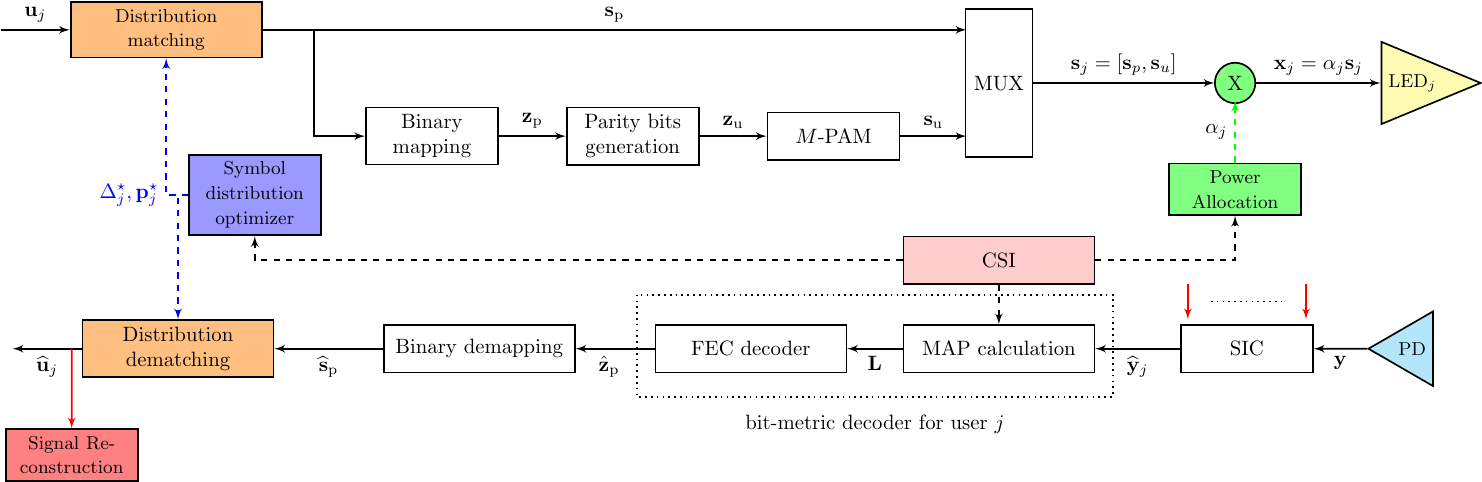}
	\caption{Encoder and decoder for user $j$ of the proposed \ac{NOMA} \ac{VLC} scheme with \ac{PS}.}
	 \label{fig:ch4:scheme}
\end{figure*}

This section presents the system model for the uplink \ac{NOMA} transmission scenario. \cref{fig:ch4:NOMAscenario} illustrates the depicted scenario, where $\Nled$ users transmit data to a common \ac{PD} at the receiver side. For data transmission, each user is equipped with a single \ac{LED}. The users are ordered based on their optical channel qualities such that 
\begin{align} \label{assumption:ChnOrder}
    h_{1}\leq h_{2}\leq\cdots\leq h_{\Nled} \mbox{,}
\end{align}
where $h_{j}$ represents the real-valued gain of the \ac{VLC} channel between the $j$-th \ac{LED} and the \ac{PD} \cite{yin2015performance} \footnote{While considering a \ac{LOS} path gain akin to \cite{fath2012performance}, our proposed system is versatile and can encompass broader channel characteristics, including \ac{NLOS} links as in \cite{mesleh2011optical}.}
\begin{equation}
h_{j}\!\!=\!\!\left\{\begin{array}{ll}
\!\!\!\!\frac{(m+1) A}{2 \pi d_{j}^{2}}\!\cos^{m}\!\left(\!\phi_{j}\!\right)\!T\!\left(\!\psi_{j}\!\right)\!G\!\left(\!\psi_{j}\!\right)\!\cos\!\left(\!\psi_{j}\!\right)\!, & \!\!0\!\leq\!\psi_{j}\! \leq \!\Psi_{c} \\
0, & \psi_{j}>\Psi_{c}
\end{array}\right.\!\!\!\!\mbox{,}
\end{equation}
where $d_{j}$ denotes the distance from the \ac{PD} to $j$-th \ac{LED}; $A$ and $\Psi_{c}$ are the area and field-of-view of \ac{PD}, respectively; $\psi_{j}$ and $\phi_{j}$ are the angles of arrival and departure from $j$-th \ac{LED} to the \ac{PD} with respect to their normal axes, respectively; $T\!\left(\psi_{j}\right)$ is the gain of the optical filter; $G\!\left(\psi_{j}\right)$ denotes the gain of the optical concentrator \cite{yin2015performance}; $\Phi_{1/2}$ is the half-power angle of the \ac{LED}; and the Lambert's mode number is 
\begin{equation}
    m=\frac{-\ln{(2)}}{\ln{\left(\cos{\left(\Phi_{1/2}\right)}\right)}} \mbox{.}
\end{equation}

In the \ac{NOMA} scheme, the received signal $y$ at the \ac{PD} is formed by superposition of all signals transmitted from the \acp{LED} in the power domain, i.e.,
\begin{align} \label{eq:ch4:SignalModelPD}
y=\sum\limits_{j=1}^{\Nled}h_{j}x_{j}+w=\sum\limits_{j=1}^{\Nled}h_{j}\alpha_{j}s_{j}+w \mbox{,}
\end{align}
 where $x_{j}\triangleq \alpha_{j} s_{j}$ is the transmitted optical intensity symbol from user $j$; $s_{j} \in \mathcal{S}_{j}\triangleq\{\Delta_{j},2\Delta_{j},\cdots,M\Delta_{j}\}$ is the unipolar $M$-ary \ac{PAM} constellation symbol of user $j$; $\Delta_{j}$ is the spacing between adjacent symbols of user $j$; $w$ is the sum of thermal and high intensity ambient light shot noises, which can be modeled as an \ac{AWGN} with mean zero and variance $\sigma^{2}$ \cite{mesleh2011optical}; and $\alpha_{j}\triangleq1/h_{j}$ is
the power allocation weight for the $j$-th user \cite{yin2015performance}.

Building upon the signal model, we now introduce the constraints that govern the power allocation in the \ac{NOMA} scheme. To begin, we introduce the notations $S_{j}$ and $X_{j}$ as \acp{r.v.} representing the constellation and transmitted symbols of user $j$, respectively. The $i$-th constellation symbol of user $j$ is denoted as $s_{ji}=\Delta_{j}\;\!\!i$, where $i\in \mathbb{M}\triangleq\{1,2,\cdots,M\}$, and $x_{ji}\triangleq \alpha_{j} s_{ji}$. Furthermore, $p^{i}_{j}$ is defined as the probability of transmitting the $i$-th constellation symbol from user $j$, i.e.,
\begin{equation}
p^{i}_{j} \triangleq \mathbb{P}\{S_{j}=s_{ji}\}=\mathbb{P}\{X_{j}=x_{ji}\}\mbox{.}
\end{equation}
Therefore, the average optical transmitted power constraint at user $j$ is
\begin{align}
    \gamma_{j}\triangleq\alpha_{j}\sum\limits_{i=1}^{M}s_{ji}p^{i}_{j}=\alpha_{j}\;\mathbf{s}^{T}_{j}\;\mathbf{p}_{j}=\alpha_{j}\;P_{rj}\leq P_{\max}\mbox{,}
    \label{eq:ch4:PowerConstraint}
\end{align}
where $\mathbf{p}_{j}\triangleq\left[p^{1}_{j},p^{2}_{j},\cdots,p^{M}_{j}\right]^{T}$ is the \ac{PMF} of $S_{j}$; $\mathbf{s}_{j}\triangleq\left[s_{j1},s_{j2}\cdots,s_{jM}\right]^{T}$; $P_{rj}$ is a target average optical received power from user $j$; and $P_{\max}$ is the maximum average optical transmit power.

From \eqref{eq:ch4:PowerConstraint}, the constellation spacing that satisfies the power constraint is
\begin{align}\label{constraint:ch4:Delta}
     \frac{P_{rj}}{M}\leq\Delta_{j}\!\leq\!P_{rj} \mbox{.}
\end{align}
 This work utilizes the instantaneous \ac{OSNR} which is expressed as $P_{rj}/\sigma$ instead of the electrical \ac{SNR}, as \ac{OSNR} is more pertinent for \ac{VLC} systems \cite{farid2009channel}.


In the \ac{NOMA} system architecture depicted in \cref{fig:ch4:NOMAarchitectue}, the receiver utilizes \ac{SIC} detection. The decoding order at the receiver is determined based on the average received signal strength, which is influenced by the channel gains. An efficient uplink power control design ensures that the average optical received powers of different users are well separated by allocating larger signal powers to users with lower channel gains \cite{zhang2016uplink}. The average optical received power from user $j$ can be expressed as 
\begin{align}
    P_{rj}=P_{rN}\; 10^{\frac{\left(N-j\right)\Pbcf}{10}}, \quad \forall j \!\in\!\mathbb{N}\triangleq\{1, 2,\cdots, N\}\mbox{,} 
\end{align}
where
\begin{align}\label{eq:ch4:PowerbackOff}
     \Pbcf=\frac{10\log\left(\frac{P_{\max}}{\alpha_{1}P_{rN}}\right)}{N-1}
 \end{align}
 is a power back-off coefficient that guarantees the satisfaction of the average optical transmitted power constraint for each user. Additionally, it allows the average optical received power from user $j$ signal to be $\Pbcf-$dB weaker than the power from user $(j-1)$ signal, enabling successive mitigation of co-channel interferences. 
 
 In \ac{NOMA} with \ac{SIC}, the decoding process for user $j$ involves demodulating and decoding signals from users $1$ to $j-1$, re-encoding and re-modulating the recovered information bits at the receiver from these users, and subtracting them from the superimposed received signal at user $j$. The remaining signals from users $j+1$ to $N$ are considered as inter-user interference. To ensure perfect interference cancellation with \ac{SIC}, the probability of error in detection should be exceptionally low. This can be achieved by operating within an appropriate \ac{AR} of the \ac{OMAC} channel and deploying robust channel coding with an extended code length. Consequently, the received signal at the user $j$ decoder, post \ac{SIC} processing, can be accurately approximated as

\begin{align}\label{eq:ch4:RxSignalj}
    \widehat{y}_{j}\triangleq y-\sum\limits_{k=1}^{j-1}h_{k}x_{k}=h_{j}x_{j}+\sum\limits_{k=j+1}^{\Nled}h_{k}x_{k}+w \mbox{.}
\end{align}

\section{Proposed NOMA VLC scheme with PS}\label{sec:ch4:ProposedScheme}
The proposed \ac{NOMA} \ac{VLC} scheme with \ac{PS} aims to enhance the \ac{SE} of \ac{NOMA} \ac{VLC} systems by probabilistically shaping the distribution of the input symbols and optimizing the constellation spacing. In this section, the proposed encoder and decoder are described.

\subsection{Encoder}\label{subsec:ch4:encoder}
The encoder for user $j$ of the $\scheme$ scheme is shown in \cref{fig:ch4:scheme}. The encoder transforms information bits into probabilistically shaped channel-coded symbols. In our proposed encoder, a reverse concatenation architecture is used, where the distribution matching is executed before the channel encoding, in order to mitigate the common problem of error bursts when one corrupted symbol at the dematcher input can lead to various corrupted bit-errors at the dematcher output. The proposed encoder consists of several components, including \ac{DM}, binary mapper, parity bits generator of the \ac{FEC} encoder, multiplexer, $M$-\ac{PAM} converter, symbol distribution optimizer, and \ac{LED}.

\subsubsection{Symbol Distribution Optimizer}\label{subsubsec:ch4:DistrOptimizer}
At first, the capacity-achieving distribution of the proposed scheme defined by parameters $\mathbf{p}^{\star}_{j}$ and $\Delta^{\star}_{j}$ is computed. This distribution allows us to achieve the highest reliable communication rate given a specific \ac{OSNR}. In \cref{sec:ch4:RateAdaptation}, we present an algorithm for obtaining the capacity-achieving distribution.

\subsubsection{Distribution Matcher}\label{subsubsec:ch4:DM}
 In the encoder, an invertible fixed-length  \ac{DM} is used, i.e.,  \ac{CCDM}. All possible output sequences of \ac{CCDM} have an identical target distribution $\mathbf{p}_{j}$. In general, the goal of \ac{CCDM} is to output symbols with distribution $\mathbf{p}_{j}$ from uniformly distributed stream of input bits. The rate of \ac{CCDM} is $\mathrm{R}_{\mathrm{ccdm}}\triangleq k/n_{\text{p}}$, where $k$ and $n_{\text{p}}$ are the number of input bits and output symbols, respectively. For an output sequence length that tends towards infinity, the \ac{CCDM} rate converges to the source entropy, i.e.,
\begin{align}
    \underset{n_{p} \to \infty}{\lim} \frac{k}{n_{p}}=\mathcal{H}\left(S_{\text{p}}\right) \mbox{,}
\end{align}
where $\mathcal{H}\left(S_{\text{p}}\right)=-\sum\limits_{i=1}^{M}p^{i}_{j} \log\left(p^{i}_{j} \right)$ is the entropy of the discrete \ac{r.v.} $S_{\text{p}}$, representing the shaped symbols at the output of \ac{CCDM}. In case of finite $n_{\text{p}}$, the \ac{CCDM} rate is less than the source entropy. However, the rate loss is negligible for a long finite output sequence length \cite{schulte2015constant}. For example, at $n_{\text{p}}=64800$, the rate loss is less than $7.5\times 10^{-4}$ bits/symbol for $8$-\ac{PAM}  \cite{Dvbs2x:15}. 
 
In our proposed encoder, the \ac{CCDM} transforms uniformly distributed stream of information bits, $\mathbf{u}_{j}\in\{0,1\}^{k}$, into probabilistically shaped unipolar $M$-\ac{PAM} symbols represented by $\mathbf{s}_{\text{p}}\in\{\Delta^{\star}_{j}, 2\Delta^{\star}_{j},\cdots, M\Delta^{\star}_{j}\}^{n_{\text{p}}}$ with the target distribution $\mathbf{p}^{\star}_{j}$, which is computed in \cref{sec:ch4:RateAdaptation}.

\subsubsection{Channel Encoder}
Binary \ac{FEC} encoder is being used to achieve reliable communication with a data rate that is close to the channel capacity. The probabilistically shaped $M$-\ac{PAM} symbols, $\mathbf{s}_{\text{p}}\triangleq\left[{s}^{1}_{\text{p}},{s}^{2}_{\text{p}},\cdots,{s}^{n_{\text{p}}}_{\text{p}}\right]$, are converted into a vector of bits, $\mathbf{z}_{\text{p}}\in\{0,1\}^{\log\left(M\right)\;n_{\text{p}}}$, using the binary mapper such that 
\begin{align}
    \mathbf{z}_{\text{p}}\triangleq \left[\mathbb{B}\left({s}^{1}_{\text{p}}\right),\mathbb{B}\left({s}^{2}_{\text{p}}\right),\cdots,\mathbb{B}\left({s}^{n_{\text{p}}}_{\text{p}}\right) \right]\mbox{,}
\end{align}
where
\begin{align}
    \mathbb{B}\!\left({s}^{k}_{\text{p}}\right)\!\triangleq\!\left[b^{k}_{1},b^{k}_{2},\cdots\!,b^{k}_{\log\left(M\right)}\right]\!\!, \,\,\text{for } k\in \{1,2,\cdots,n_{ p}\},
\end{align}
and $b^{k}_{l}$ is the bit level $l$ of the symbol ${s}^{k}_{\text{p}}$. We use Gray code binary mapper, as it yields good performance \cite{bocherer2015bandwidth}. 

Then, the bits are coded with the systematic binary \ac{FEC} encoder with a channel coding rate $\mathrm{R}_{\text{FEC}}$ as 
\begin{align}
    \mathbf{z}^{T}_{\text{u}}=\Pmatrix\;\mathbf{z}^{T}_{\text{p}}\mbox{,}
\end{align}
where $\Pmatrix\in\{0,1\}^{\left(\frac{1}{\Rfec}-1\right) n_{\text{p}} \log\left(M\right)\times n_{\text{p}} \log\left(M\right)}$ is the parity matrix, which can be obtained by representing the code generator matrix in a standard form \cite{biglieri2005coding}. Please note that, regardless of the distribution of $\mathbf{z}_{\text{p}}$, the parity bits, $\mathbf{z}_{\text{u}}\in\{0,1\}^{\left(\frac{1}{\Rfec}-1\right) n_{\text{p}} \log\left(M\right)}$, are approximately uniformly distributed. The reason is that when a large number of bits are added, the distribution of the resulting bits tends to be uniform, even if the distribution of the original bits is not uniform \cite{bocherer2019probabilistic}. 

\subsubsection{Sparse-Dense Signal}
The uniformly distributed unipolar $M$-\ac{PAM} symbols, $\mathbf{s}_{\text{u}}\in\{\Delta^{\star}_{j}, 2\Delta^{\star}_{j},\cdots, M\Delta^{\star}_{j}\}^{\left(\frac{1}{\Rfec}-1\right) n_{\text{p}}}$, are produced from the parity bits, $\mathbf{z}_{\text{u}}$, using the binary demapper. After that, the multiplexer appends probabilistically shaped and uniformly distributed symbols to form a codeword consisting of $n=n_{\text{p}}/\Rfec$ symbols, such that  
\begin{align}\label{eq:ch4:codeword}
    \mathbf{s}_{j}=\left[\;\mathbf{s}_{\text{p}}\;,\; \mathbf{s}_{\text{u}}\;\right]\mbox{.}
\end{align}
Consequently, the codeword in \eqref{eq:ch4:codeword} is multiplied by the power allocation coefficient $\alpha_{j}$ to form the transmitted signal, i.e.,
\begin{align}\label{eq:ch4:signalTX}
    \mathbf{x}_{j}=\alpha_{j}\mathbf{s}_{j}\mbox{.}
\end{align}
The signal in \eqref{eq:ch4:signalTX} is a \ac{SDT} signal because uniformly distributed symbols have the maximum source entropy (representing dense information), while the probabilistically shaped symbols represent smaller entropy (representing sparse information).

\subsection{Decoder}
After discussing the proposed encoder and its output in the previous subsection, our attention now shifts to the decoder for user $j$ within the $\scheme$ scheme, as illustrated in \cref{fig:ch4:scheme}. The objective is to achieve the maximum \ac{AR} of the \ac{SDT} signal by employing the optimal \ac{SMD} for \ac{FEC} decoding. However, it is worth noting that \ac{SMD} entails high computational complexity \cite{bocherer2015bandwidth}. Therefore, to strike a balance between computational complexity and achieving a rate close to the \ac{SDT} capacity, we propose the utilization of \ac{BMD} with soft decision-making capabilities. This approach involves calculating the \ac{MAP} for each bit level of every received symbol. The selection of \ac{MAP} is motivated by its incorporation of prior knowledge, including the probability distribution of the transmitted symbols. In contrast, \ac{ML} decoding relies solely on the received signal and assumes a uniform distribution of the transmitted symbols.

\subsubsection{\ac{SIC}}
As can be seen from \cref{fig:ch4:scheme}, the received signal $\mathbf{y}$ at the \ac{PD} is the input to \ac{SIC} block. After applying \ac{SIC} as described in \eqref{eq:ch4:RxSignalj}, the output of \ac{SIC} block is $\widehat{\mathbf{y}}_{j}\triangleq\left[\widehat{y}^{1}_{j},\widehat{y}^{2}_{j},\ldots, \widehat{y}^{n}_{j}\right]$, where $\widehat{y}^{t}_{j}$ is the received signal at the decoder of user $j$ at time instant $t\in\{1, 2, \ldots, n\}$.

\subsubsection{MAP Calculation}
Let us consider the $M$-\ac{PAM} symbol transmitted by user $j$ at time instant $t$, denoted by $x^{t}_{j}$. It consists of individual bit levels represented by \acp{r.v.} $B^{t}_{j,l}\in\{0,1\}$, where $l$ denotes the bit level. Therefore, its binary representation is
\begin{align}
    \mathbb{B}\left(x^{t}_{j}\right)=\left[B^{t}_{j,1}, B^{t}_{j,2}, \cdots, B^{t}_{j,\log\left(M\right)}\right]\mbox{.}
\end{align}
The \acp{r.v.} representing $x^{t}_{j}$ and $\widehat{y}^{t}_{j}$ are $X^{t}_{j}$ and $\widehat{Y}^{t}_{j}$, respectively. Therefore, the \ac{MAP} of the $l$th bit level given $\widehat{y}^{t}_{j}$ is 
\begin{align}\label{eq:ch4:LLR}
L^{j}_{t,l}&=\log\left(\frac{f_{B^{t}_{j,l}|\widehat{Y}^{t}_{j}}\left(0|\widehat{y}^{t}_{j}\right)}{f_{B^{t}_{j,l}|\widehat{Y}^{t}_{j}}\left(1|\widehat{y}^{t}_{j}\right)}\right) \nonumber \\
    &=\log\left(\frac{\sum\limits_{x_{j} \in \mathbb{X}_{l}^{0}}f_{\widehat{Y}^{t}_{j}|X^{t}_{j}}\left(\widehat{y}^{t}_{j}|x\right)\mathbb{P}\{X^{t}_{j}=x_{j}\}}{\sum\limits_{x_{j} \in \mathbb{X}_{l}^{1}}f_{\widehat{Y}^{t}_{j}|X^{t}_{j}}\left(\widehat{y}^{t}_{j}|x\right)\mathbb{P}\{X^{t}_{j}=x_{j}\}}\right) \mbox{,}
\end{align}
where $\mathbb{X}_{l}^{0}$ and $\mathbb{X}_{l}^{1}$ are the sets containing all the values of $x_{j}$ whose $l$th bit level is $0$ and $1$, respectively. Please note that, in \ac{SDT} signal, the distribution of $X^{t}_{j}$ depends on $t$. Hence, for $t\in\{n\Rfec+1,n\Rfec+2,\ldots,n\}$, $\mathbb{P}\{X^{t}_{j}=x_{j}\}=1/M$. 

\subsubsection{FEC Decoder and Dematcher}
Since $L^{j}_{t,l}$ serves as the sufficient statistic for estimating the bit level $l$ based on the received symbol $\widehat{y}^{t}_{j}$ \cite{bocherer2015bandwidth}, the \ac{FEC} soft decoder takes \ac{MAP} values obtained in \eqref{eq:ch4:LLR} as input for error correction. Consequently, the output of the \ac{FEC} soft decoder is denoted as $\hat{{\mathbf z}}_{\text p}\in \{0,1\}^{\log(M)n_{ p}}$. Following this, the corresponding unipolar $M$-\ac{PAM} symbols, $\widehat{{\mathbf s}}_{\text p}\in\{\Delta^{\star}_{j}, 2\Delta^{\star}_{j},\cdots, M\Delta^{\star}_{j}\}^{n_{\text{p}}}$, are obtained from $\hat{{\mathbf z}}_{\text p}$ using a binary demapper. Lastly, the distribution dematching process maps the symbols from $\widehat{{\mathbf s}}_{\text p}$ to the corresponding recovered information bits of user $j$, denoted as $\widehat{\mathbf u}_{j}\in \{0,1\}^{k}$, which are subsequently re-encoded and re-modulated for utilization in \ac{SIC} operations.

\section{Rate Analysis}\label{sec:ch4:CapAnalysis}
To find the capacity-achieving distribution of the proposed scheme, we need to derive its \ac{AR} and \ac{TR}. \ac{AR} of the proposed scheme is derived accounting for \textit{i)} constraints imposed by the \ac{IM/DD} such as non-negative and real-valued input signal; \textit{ii)} \ac{NOMA} features such as \ac{OMAC} properties and employment of the \ac{SIC} detection; \textit{iii)} the \ac{SDT} signalling in \eqref{eq:ch4:signalTX}, where only part of the transmitted symbols can be probabilistically shaped. 

 \subsection{Achievable Rate of the Proposed Scheme}
 In this subsection, we derive the \ac{AR} of the proposed scheme. Let $\widehat{Y}_{j}$ be a \ac{r.v.} representing the received signal in \eqref{eq:ch4:RxSignalj}. According to the \ac{SIC} decoding principle in \ac{NOMA}, the received signal at the decoder of user $j$ (i.e., $\widehat{y}_{j}$) includes interference signals from users $j+1$ to $N$, which have constellation spacing values denoted by $\DeltaIntf\triangleq\left[\Delta_{j+1},\Delta_{j+2},\cdots,\Delta_{N}\right]$ and follow \acp{PMF} denoted by $\PdIntrf\triangleq\left[\mathbf{p}^{T}_{j+1},\mathbf{p}^{T}_{j+2},\cdots,\mathbf{p}^{T}_{N}\right]$, respectively. We use $\AIntrf\triangleq\left[\DeltaIntf,\PdIntrf\right]$ to represent the spacing values and \acp{PMF} of the interference signals in $\widehat{y}_{j}$.  

An \ac{AR} of user $j$ of the proposed scheme is 
\begin{align}
\!\Rsdt&\!\triangleq\!\Rfec\Roimac\!\nonumber \\
&\phantom{aaa}+\!\left(1\!-\!\Rfec\right)\Ru\!\!\mbox{,}
\end{align}
where $\mathbf{p}^{u}_{j}\triangleq\left[1/M, 1/M, \ldots, 1/M\right]^{T}$ and  
\begin{align}
\Roimac&= \mathcal{I}\left(X_{j}; \widehat{Y}_{j}\right)\nonumber \\
&=\mathcal{H}\left(X_{j}\right)-\mathcal{H}\left(X_{j} | \widehat{Y}_{j}\right) \nonumber \\
   &=-\sum_{i=1}^{M} p^{i}_{j}\log\left( p^{i}_{j}\right)-\mathcal{H}\left(X_{j} | \widehat{Y}_{j}\right),
\end{align}
where $\mathcal{I}\left(X_{j}; \widehat{Y}_{j}\right)$ represents the \ac{MI} between transmitted signal $X_{j}$ and received signal $\widehat{Y}_{j}$, $\mathcal{H}\left(X_{j}\right)$ is the entropy of the \ac{r.v.} $X_{j}$, and $\mathcal{H}\left(X_{j} | \widehat{Y}_{j}\right)$ is the conditional entropy, which can be derived as
\begin{align}\label{eq:ch4:H(X|Y)_entr}
    \mathcal{H}\!\!\left(\!X_{j} | \widehat{Y}_{j}\!\right)&\!\!=\!\!-\!\!\!\!\int\limits_{-\infty}^{\infty}\!\!\!f_{\!\widehat{Y}_{j}}\!(\widehat{y}_{j})\!\sum_{i=1}^{M} f_{\!\!X_{j}|\widehat{Y}_{j}}\!\!\left(s_{ji}|\widehat{y}_{j}\right)\!\log\!\left(\!f_{\!X_{j}|\widehat{Y}_{j}}\!\left(s_{ji}|\widehat{y}_{j}\right)\!\right)\!\mathrm{d}\widehat{y}_{j} \nonumber \\
    &\!\!=\!\!-\!\!\!\!\int\limits_{-\infty}^{\infty}\!\!\!\! f_{\widehat{Y}_{j}}\!\!(\widehat{y}_{j})\!\!\sum_{i=1}^{M} \frac{f_{\widehat{Y}_{j}|X_{j}}(\widehat{y}_{j}|s_{ji}) f_{X_{j}}(s_{ji})}{f_{\widehat{Y}_{j}}(\widehat{y}_{j})} \nonumber \\
    &\phantom{===}\times\log\left(\frac{f_{\widehat{Y}_{j}|X_{j}}(\widehat{y}_{j}|s_{ji}) f_{X_{j}}(s_{ji})}{f_{\widehat{Y}_{j}}(\widehat{y}_{j})} \right)\mathrm{d}\widehat{y}_{j} \nonumber \\
    &\!\!=\!\!-\!\!\!\!\int\limits_{-\infty}^{\infty}\!\! \sum_{i=1}^{M} f_{\widehat{Y}_{j}|X_{j}}(\widehat{y}_{j}|s_{ji}) f_{X_{j}}(s_{ji}) \nonumber \\
    &\phantom{===}\times\log\left(\frac{f_{\widehat{Y}_{j}|X_{j}}(\widehat{y}_{j}|s_{ji}) f_{X_{j}}(s_{ji})}{f_{\widehat{Y}_{j}}(\widehat{y}_{j})} \right)\mathrm{d}\widehat{y}_{j} \nonumber \\
    &\!\!=\!\!\int\limits_{-\infty}^{\infty}\!\! \sum_{i=1}^{M} f_{\widehat{Y}_{j}|X_{j}}(\widehat{y}_{j}|s_{ji}) f_{X_{j}}(s_{ji}) \nonumber \\
    &\phantom{===}\times\log\left(\frac{f_{\widehat{Y}_{j}}(\widehat{y}_{j})}{f_{\widehat{Y}_{j}|X_{j}}(\widehat{y}_{j}|s_{ji}) f_{X_{j}}(s_{ji})} \right)\mathrm{d}\widehat{y}_{j} \mbox{,}
\end{align}
where 
\begin{align}
f_{\widehat{Y}_{j}|X_{j}}\!\!\!\left(\widehat{y}_{j}|s_{ji}\right)\!\!&=\!\!\!\!\!\!\!\!\!\!\!\!\!\!\!\!\!\!\!\!\!\!\!\sum\limits_{\{s_{k}\}_{k=j+1}^{N} \in \{\mathcal{S}_{k}\}_{k=j+1}^{N}} \!\!\!\!\!\!\!\!\!\!\!\!\!\!\!\!\!\!\!\!f_{\widehat{Y}_{j}|\{X_{k}\}_{k=j}^{N}}\!\!\left(\widehat{y}_{j}|s_{ji},\{s_{k}\}_{k=j+1}^{N}\right)\!\!\!\!\!\prod\limits_{k=j+1}^{N} \!\!\!\!f_{X_{k}}\!\!\left(s_{k}\right)
    \nonumber \\ 
    &=\!\!\!\!\!\!\!\!\!\!\!\!\!\!\!\!\!\!\!\!\!\!\sum\limits_{\{s_{k}\}_{k=j+1}^{N} \in \{\mathcal{S}_{k}\}_{k=j+1}^{N}}\!\!\!\!\!\!\!\!\!\!\!\!\!\!\!\!\!\!\!\!\!\frac{e^{\left(\!\frac{-\left(\widehat{y}_{j}\!-\!h_{j}\alpha_{j}s_{ji}-\!\!\!\sum\limits_{k=j+1}^{N}\!\!\!\!\!h_{k}\alpha_{k}s_{k}\right)^{2}}{2\sigma^{2}}\!\right)}}{\sqrt{2\pi\sigma^{2}}}\!\!\!\!\!\prod\limits_{k=j+1}^{N} \!\!\!\!f_{X_{k}}\!\!\left(s_{k}\right)\!\mbox{,}
\end{align}
and where
\begin{align}
    f_{\widehat{Y}_{j}}\left(\widehat{y}_{j}\right)&=\sum_{i^{\prime}=1}^{M} f_{\widehat{Y}_{j}|X_{j}}(\widehat{y}_{j}|s_{ji^{\prime}}) f_{X_{j}}(s_{ji^{\prime}})\mbox{,}
\end{align}
\begin{align}
    f_{X_{j}}(s_{ji})=p^{i}_{j}\mbox{.}
\end{align}

Note that $\Roimac$ is the \ac{AR} of user $j$ for \ac{NOMA} with \ac{IM/DD}, where all of the transmitted symbols can be probabilistically shaped.

\subsection{Transmission Rate of the Proposed Scheme}\label{subsec:ch4:RatesProposedScheme}
We now shift our focus to deriving the \ac{TR} achieved by the $\scheme$ scheme. The \ac{TR} is defined as a fraction of information bits over the number of channel uses. The \ac{TR} of user $j$ for the proposed scheme shown in \cref{fig:ch4:scheme} can be derived as
\begin{align}
    \!\!\mathrm{T}_{j}\!\left(\mathbf{p}_{j}\right) 
    \!\triangleq\!\frac{k}{n} 
    &\!=\!\frac{n_{p}\mathcal{H}\left(S_{\text{p}}\right)}{n}\!=\!-\Rfec \sum\limits_{i=1}^{M}p^{i}_{j} \log\left(p^{i}_{j} \right)\!\mbox{,} 
    \label{eq:ch4:Txrate}
\end{align}
where $k$ is the number of information bis, and $n$ is the frame size. Note that the following condition $\mathrm{T}_{j}\leq\mathrm{R}^{\mathrm{SDT}}_{j}$ should be satisfied in order to guarantee reliable communication. When this condition is satisfied, the \ac{TR} of the proposed scheme is achievable.

\section{Rate Adaptation}\label{sec:ch4:RateAdaptation}
 The TR of the $\scheme$ scheme for user $j$ can be adapted by adjusting the distribution of unipolar $M$-PAM symbols, denoted as $\mathbf{p}_j$, and the channel coding rate, denoted as $\Rfec$. In this section, we propose an algorithm to maximize the achievable \ac{TR} by optimizing 
 $\mathbf{p}_{j}$ and $\Delta_{j}$ for given \ac{OSNR} and $\Rfec$.

\subsection{Optimization Problem}\label{subsec:ch4:optimizationTx}
In this subsection, we delve into the rate adaptation problem by presenting an optimization problem that allows us to mathematically formulate and seek the optimal $\mathbf{p}_{j}$ and $\Delta_{j}$. To determine the \ac{TR} of each user, it is necessary to jointly optimize $\mathbf{p}_{j}$ and $\Delta_{j}$ for every user. Consequently, there are $N$ separate optimization problems that need to be solved sequentially.

Considering the power and rate constraints, the maximum achievable \ac{TR} of user $j$ of the $\scheme$ scheme for the given $\Rfec$ and \ac{OSNR} can be obtained from 
\begin{subequations}\label{optimization:ch4:MI:Tx}
\begin{alignat}{2}
&\underset{\Delta_{j},\mathbf{p}_{j}}{\operatorname{maximize}} \quad && \mathrm{T}_{j}\left(\mathbf{p}_{j}\right) \\
&\text{subject to} \quad &&\Rfec\;\mathbf{s}^{T}_{j}\;\mathbf{p}_{j}+\left(1-\Rfec\right)\mathbf{s}^{T}_{j}\;\mathbf{p}^{u}_{j}=P_{rj}, \label{eq:ch4:Power:constraint1}\\
& \quad && \mathrm{T}_{j}\left(\mathbf{p}_{j}\right) \leq \Rsdtopt-\Rbf ,\label{eq:ch4:TX:constraint}\\
&\quad && \Delta_{j,\min}\leq\Delta_{j}\!\leq\!\Delta_{j,\max},\label{eq:ch4:Deltaj:constraint}\\
& \quad &&  \sum_{i=1}^{M} p^{i}_{j}=1,  \\
& \quad &&   p^{i}_{j} \geq 0, \quad \forall i \in\mathbb{M}  \mbox{,}
\end{alignat}
\end{subequations}
where $\AIntrf^{\star}\triangleq\left[\DeltaIntf^{\star},\PdIntrf^{\star}\right]$, $\Rbf$ is the back-off rate to compensate the rate loss due to the use of \ac{BMD} and finite-length distribution matcher and 
\begin{align}
\Delta_{j,\min}&\triangleq \frac{P_{rj}}{\left(\Rfec M+\left(1-\Rfec\right)\frac{M+1}{2}\right)}\mbox{,}\\ 
\Delta_{j,\max}&\triangleq \frac{P_{rj}}{\left(\Rfec+\left(1-\Rfec\right)\frac{M+1}{2}\right)}\mbox{.}
\end{align}
The constraint \eqref{eq:ch4:Power:constraint1} ensures that the total power, taking into account the coding rate and symbol distribution, is equal to the received power. The minimum and maximum symbol spacings, denoted as $\Delta_{j,\min}$ and $\Delta_{j,\max}$, are derived based on the power allocation and coding rate to satisfy the power constraint. The constraint \eqref{eq:ch4:TX:constraint} guarantees reliable communication with an achievable \ac{TR}.

In \ac{NOMA} with \ac{SIC}, the decoding order starts from user $1$ and ends at user $N$. The received signal of the last decoded user is given by $\widehat{y}_{N}=h_{N}x_{N}+w$, where there is no interference from other users. For the received signal of the prior decoded user (i.e., $\widehat{y}_{j<N}$), we observe the presence of interference from both \ac{AWGN} and multiuser interference from users $j+1$ to $N$. To ensure optimal
distribution of the user $j<N$ signal, it is essential to account for this inter-user interference. Therefore, the order of optimization is opposite to the \ac{SIC} decoding order, which means that the optimization problem \eqref{optimization:ch4:MI:Tx} is solved starting from user $N$ and then progressing backward to user $1$. For example, we first find $\mathbf{p}^{\star}_{N}$ and $\Delta^{\star}_{N}$ for user $N$ from \eqref{optimization:ch4:MI:Tx}, where $\mathbf{a}^{\star}_{\widehat{y}_{N}}$ is an empty vector. Next, we obtain $\mathbf{p}^{\star}_{N-1}$ and $\Delta^{\star}_{N-1}$ for user $N-1$ from \eqref{optimization:ch4:MI:Tx}, taking into account the interference signal from user $N$ in $\widehat{y}_{N-1}$, which follows the distribution given by $\mathbf{p}^{\star}_{N}$ and has constellation spacing $\Delta^{\star}_{N}$, i.e., $\mathbf{a}^{\star}_{\widehat{y}_{N-1}}=\left[\Delta^{\star}_{N},\mathbf{p}^{\star T}_{N}\right]$. We continue this process to obtain the optimal \ac{PMF} and spacing for the remaining users in accordance with the optimization order.

\subsection{Proposed Algorithm}
This subsection proposes an algorithm to efficiently find the capacity-approaching distribution of the proposed scheme. Optimizing $\mathbf{p}_{j}$ and $\Delta_{j}$ jointly in \eqref{optimization:ch4:MI:Tx} is excessively complex non-convex problem. Therefore, we propose an iterative approach where one variable is optimized while the other is fixed alternately. Nevertheless, optimizing $\mathbf{p}_{j}$ while fixing $\Delta_{j}$ results in the non-convexity of the optimization problem due to \eqref{eq:ch4:TX:constraint}. To address this challenge, we employ a surrogate function, a mathematical construct designed to approximate the original non-convex objective function with a convex counterpart, inspired by the work of \cite{vontobel2008generalization}. This surrogate function plays a crucial role in convexifying our problem, allowing the algorithm to iteratively approach the optimized solution. With each iteration, our approach converges toward a solution that balances computational efficiency and optimality. In this regard, $\Rsdtopt$ can be rewritten as 
\begin{align}
\Rsdtopt\!\!&=\!\!\Rfec\!\left(g_1\!\left(\mathbf{p}_{j}\right)\!-\!g_2\!\left(\Delta_{j},\mathbf{p}_{j},\AIntrf^{\star}\right)\!\right) \nonumber\\
   &\phantom{==}+\!\left(1-\Rfec\right)\!\!\Ruopt\!\!\mbox{,}
\end{align}
where 
\begin{align}\label{eq:ch4:g1}
    g_1\left(\mathbf{p}_{j}\right)\!\triangleq-\sum_{i=1}^{M} p^{i}_{j}\log\left( p^{i}_{j}\right)\mbox{,}
\end{align}
\begin{align}\label{eq:ch4:g2}
    g_2\!\left(\Delta_{j},\mathbf{p}_{j},\AIntrf^{\star}\right)\triangleq \sum_{i=1}^{M} p^{i}_{j}\,  \Wip\mbox{,}
\end{align}
and 
\begin{align}\label{eq:ch4:Wpi}
\Wip\!\!&\triangleq\!\!\!\!\int\limits_{-\infty}^{\infty}\!\!\!\! f_{\widehat{Y}_{j}|X_{j}}(\widehat{y}_{j}|s_{ji}) \nonumber \\
&\phantom{aaaaa}\times\log\!\!\left(\frac{f_{\widehat{Y}_{j}}(\widehat{y}_{j})}{f_{\widehat{Y}_{j}|X_{j}}(\widehat{y}_{j}|s_{ji}) p^{i}_{j}} \right)\!\!\mathrm{d}\widehat{y}_{j} \mbox{.}
\end{align}
Optimizing $\mathrm{R}^{\mathrm{SDT}}_{j}$ with respect to $\mathbf{p}_{j}$ is complex and results in a coupling problem because the function \eqref{eq:ch4:g2} cannot be represented as a linear summation of concave functions with optimization variables $p^{i}_{j}$. Thus, we introduce the surrogate function, $\Phip$, which can be used to approximate $\mathrm{R}^{\mathrm{SDT}}_{j}$ around $\mathbf{p}_{j}=\hat{\mathbf{p}}_{j}$, i.e.,
\begin{align}\label{eq:ch4:phi:p}
   \Phip&\triangleq\!\!\Rfec\!\left(g_1\!\left(\mathbf{p}_{j}\right)\!-\!g^{\prime}_2\!\left(\Delta_{j},\mathbf{p}_{j},\hat{\mathbf{p}}_{j},\AIntrf^{\star}\right)\!\right) \nonumber\\
   &\phantom{=======}+\!\left(1-\Rfec\right)\!\!\Ruopt\!\!\mbox{,}
\end{align}
where 
\begin{align}
g^{\prime}_2\!\left(\Delta_{j},\mathbf{p}_{j},\hat{\mathbf{p}}_{j},\AIntrf^{\star}\right)\triangleq \sum_{i=1}^{M} p^{i}_{j}\,  \Wiphat\mbox{.}
\end{align}
 According to the relative entropy properties \cite{vontobel2008generalization}, if $\hat{\mathbf{p}}_{j}=\mathbf{p}_{j}$, then $g^{\prime}_2\!\left(\Delta_{j},\mathbf{p}_{j},\hat{\mathbf{p}}_{j},\AIntrf^{\star}\right)=g_2\!\left(\Delta_{j},\mathbf{p}_{j},\AIntrf^{\star}\right)$, otherwise $g^{\prime}_2\!\left(\Delta_{j},\mathbf{p}_{j},\hat{\mathbf{p}}_{j},\AIntrf^{\star}\right)> g_2\!\left(\Delta_{j},\mathbf{p}_{j},\AIntrf^{\star}\right)$. Thus, given $\mathrm{R}_{\mathrm{FEC}}$, $\Rbf$, initial \ac{PMF} $\hat{\mathbf{p}}_{j}$, and \ac{OSNR}, the convexified version of \eqref{optimization:ch4:MI:Tx} for fixed $\Delta_{j}$ is 
 \begin{subequations}\label{optimization:ch4:MI:Tx:convex}
\begin{alignat}{2}
\!\!\!\!\mathbf{p}^{\star}_{\Delta_{j}}=&\underset{\mathbf{p}_{j}}{\operatorname{maximize}} \quad && \mathrm{T}_{j}\left(\mathbf{p}_{j}\right) \\
&\text{subject to} \quad &&\Rfec\mathbf{s}^{T}_{j}\mathbf{p}_{j}\!+\!\left(1\!-\!\Rfec\right)\mathbf{s}^{T}_{j}\mathbf{p}^{u}_{j}\!=\!P_{rj},\\
& \quad && \mathrm{T}_{j}\left(\mathbf{p}_{j}\right) \leq \Phip-\Rbf, \label{eq:ch4:TX:constr:convex} \\
& \quad &&  \sum_{i=1}^{M} p^{i}_{j}=1,  \\
& \quad &&   p^{i}_{j} \geq 0, \quad \forall i \in\mathbb{M} \mbox{.}
\end{alignat}
\end{subequations}
Since $\Phip\leq\Phi\left(\mathbf{p}_{j},\mathbf{p}_{j}\right)=\Rsdtopt$, the constraint \eqref{eq:ch4:TX:constraint} is not violated by \eqref{eq:ch4:TX:constr:convex}.

We use \ac{KKT} optimality conditions to solve the convex optimization problem \eqref{optimization:ch4:MI:Tx:convex}. The Lagrangian function can be expressed as 
\begin{align}
    L&=-\mathrm{T}_{j}\left(\mathbf{p}_{j}\right)+\!\mu\!\!\left(\sum_{i=1}^{M} p^{i}_{j}-\!\!1\!\!\right)\!+\!\eta\!\left(\Rfec\sum_{i=1}^{M} \Delta\, i\, p^{i}_{j}\!-\!P_{j}\!\!\right)\nonumber \\
    &\phantom{=====}+\!\tau\left(\mathrm{T}_{j}\left(\mathbf{p}_{j}\right)-\Phip+\!\Rbf\right)\!\mbox{,}
\end{align}
where $P_{j}\triangleq P_{rj}-\left(1\!-\!\Rfec\right)\mathbf{s}^{T}_{j}\mathbf{p}^{u}_{j}$, and the Lagrangian multipliers are denoted as $\mu$, $\eta$ and $\tau$. The solution to \eqref{optimization:ch4:MI:Tx:convex} using \ac{KKT} optimality conditions is shown in the following lemma 

\begin{lemma}[\Ac{KKT} solution to \eqref{optimization:ch4:MI:Tx:convex}]\label{Lemma1:ch4}
For the given $\Delta_{j}$, $\mathrm{R}_{\mathrm{FEC}}$, $\hat{\mathbf{p}}_{j}$, and \ac{OSNR},
\begin{align}
p^{\star i}_{\Delta_{j}}\!\!=\!\!
\begin{cases}
\frac{2^{-\eta\,\Delta\, i}}{\sum\limits_{i^{\prime}=1}^{M} 2^{-\eta\,\Delta\, i^{\prime}}}
&\!\!\!\text{if}\;\CondOnePower_{1}\;\text{and}\;\CondOneRate_{1}\;\text{are}\;\text{satisfied} \\
\frac{2^{-\eta\,\!\Delta\, i-\!\tau \Wiphat}}{\sum\limits_{i^{\prime}\!=\!1}^{M} \!\!2^{-\eta\,\!\Delta\, \!i^{\prime}\!-\!\tau \Wiprimephat}} 
&\!\!\!\text{if}\;\CondOnePower_{2}\;\text{and}\;\CondOneRate_{2}\;\text{are}\;\text{satisfied} \mbox{,}
\end{cases}
\end{align}
where
\begin{align}
    \CondOnePower_{1}\triangleq\sum\limits_{i^{\prime}=1}^{M}\left(\Rfec\Delta\, i-\!P_{j}\right)2^{-\eta\Delta i}=0 \mbox{,}
    \label{eq:ch4:LemmaPower1}
\end{align}

\begin{align}\label{eq:ch4:LemmaRate1}
    \CondOneRate_{1}\triangleq\sum_{i=1}^{M} \left[\Rfec \Wiphat-D\right]2^{-\eta\,\Delta\, i}<0 \mbox{,} 
\end{align}

\begin{align}
    \CondOnePower_{2}\triangleq\sum\limits_{i^{\prime}=1}^{M}\left(\Rfec\Delta\, i-\!P_{j}\right)2^{-\eta\,\Delta\, i-\!\tau \Wiphat}=0 \mbox{,}
    \label{eq:ch4:LemmaPower2}
\end{align}
and 
\begin{align}\label{eq:ch4:LemmaRate2}
     \CondOneRate_{2}\triangleq\!\!\!\sum_{i=1}^{M}\!\left[\Rfec \Wiphat\!-\!\!D\right]\frac{2^{-\eta\,\Delta\, i}}{2^{\!\tau \Wiphat}}\!=\!0 \mbox{,}  
\end{align}
where $D=\left(1-\Rfec\right)\Ruopt-\Rbf$, $\tau$ and $\eta$ are any positive real numbers.
\end{lemma}
\begin{proof}
The detailed derivation is shown in Appendix \ref{apendix:ch4:KKT:conditions}.
\end{proof}

\begin{algorithm}[t]
  	\caption{Proposed algorithm to find optimal $\mathbf{p}_{j}$ in \eqref{optimization:ch4:MI:Tx:convex}} 
    \label{algorithm:ch4}
  	\begin{algorithmic} [1]
  		\STATE\textbf{Inputs:} 
  		$\Delta_{j}$, $\Rfec$, $\sigma$, $\Rbf$, and $\gamma_{s}>0$  \%{\small  $\gamma_{s}$ is the stopping threshold}
  		\STATE\textbf{Initialize:} vector
  	    $\hat{\mathbf{p}}_{j}$ which satisfies constraints in \eqref{optimization:ch4:MI:Tx:convex}
  		\REPEAT 
  		\STATE \textbf{Step 1:} For each $i \in \{1, 2, \cdots, M\}$, \\ \phantom{======.}compute $W\!\left(i,\Delta_{j},\hat{\mathbf{p}}_{j},\AIntrf^{\star}\right)$
  		\STATE \textbf{Step 2:} Obtain new distribution vector $\mathbf{p}^{\star}_{\Delta_{j}}$ from Lemma \ref{Lemma1:ch4}\\
  		\STATE \textbf{Step 3:} $\epsilon=||\hat{\mathbf{p}}_{j}-\mathbf{p}^{\star}_{\Delta_{j}}||^2$
  		\STATE \textbf{Update:} $\hat{\mathbf{p}}_{j}:=\mathbf{p}^{\star}_{\Delta_{j}}$
  		
  		\UNTIL{$\epsilon <\gamma_{s}$}
  		\STATE \textbf{Output:} $\mathbf{p}^{\star}_{\Delta_{j}}$
  	\end{algorithmic}
  \end{algorithm}
The proposed algorithm leverages Lemma \ref{Lemma1:ch4} to efficiently find optimal \ac{PMF} of user $j$ for the given $\Delta_{j}$, i.e., $\mathbf{p}^{\star}_{\Delta_{j}}$. The algorithm updates $\hat{\mathbf{p}}_{j}$ in each iteration by computing the new distribution vector $\mathbf{p}^{\star}_{\Delta_{j}}$ and calculating the distance between $\hat{\mathbf{p}}_{j}$ and $\mathbf{p}^{\star}_{\Delta_{j}}$. The iteration continues until the distance falls below the threshold $\gamma_{s}$, indicating convergence. Finally, the algorithm outputs the optimal distribution vector $\mathbf{p}^{\star}_{\Delta_{j}}$.

  To determine the optimized constellation spacing, $\Delta^{\star}_{j}$, that maximizes $\Rsdtopt$ for fixed $\mathbf{p}_{j}$, we employ the golden section method. The range of search for $\Delta_{j}$ is specified in \eqref{eq:ch4:Deltaj:constraint}. \cref{algorithm:ch4} and golden section method are iterated for each of the \ac{FEC} coding rates of a \ac{LDPC} channel encoder in a DVB-S2 receiver \cite{Dvbs2x:15}. For a given \ac{OSNR}, the optimal \ac{FEC} coding rate, $\Rfec^{\star}$, is the one which results in the highest achievable \ac{TR}.

\subsection{Computational Complexity}\label{subsec:ch4:complexity}
This subsection analyzes the computational complexity of the $\scheme$ scheme. The major portion of the computational complexity results from solving the optimization problem to get the capacity-achieving distribution. 

The proposed approach adopts an iterative alternate optimization method, necessitating the execution of \cref{algorithm:ch4} for each constellation spacing value. We denote $v_{p}$ as the number of iterations required to meet the convergence criteria in \cref{algorithm:ch4}. As shown in Table \ref{table:ch4:Comp:Compelxity}, $v_{p}$ exhibits an approximately linear increase with the constellation size $M$. Additionally, $v_{\Delta}$ represents the number of iterations required to converge to the optimal constellation spacing value in the golden section algorithm.

Furthermore, we define $v_{fec}$ as the number of coding rates accommodated by the \ac{FEC} encoder for the given $M$-ary modulation. Consequently, finding the optimal parameters for $N$ users with a specific \ac{OSNR} necessitates solving $Nv_{p}v_{\Delta}v_{fec}$ convex optimization problems. The computational complexity of finding a solution to a single convex optimization problem for an $M$-ary modulation, employing the interior-point method, scales as $\mathcal{O}\left(M^{3}\right)$ \cite{potra2000interior}.

\begin{table}[t!]
  \caption{Convergence of Algorithm \ref{algorithm:ch4}}
\centering
 \begin{tabular}{||c | c||}
 \hline
 M & $v_{p}$: Number of Iterations until convergence \\ [0.5ex] 
 \hline\hline
 2 & 1\\
 \hline
  4 & 6\\
 \hline
  8 & 26\\
   \hline
  16 & 403\\
 \hline
  64 & 512\\[1ex]
 \hline
\end{tabular}
\label{table:ch4:Comp:Compelxity}
\end{table}

The optimization process can be carried out either online or offline, considering the computational resources available at the transmitting side. Offline optimization involves obtaining the optimal constellation size, \ac{FEC} rate, symbol distribution and constellation spacing for a predefined set of \acp{OSNR}. These optimized parameters are then stored in memory, which reduces the computational complexity during transmission. However, the discretization process may result in some rate loss.

\section{Numerical Results}\label{sec:ch4:PerfAnalysis}
In this section, we present numerical results to evaluate the performance of our $\scheme$ scheme in terms of \ac{SE} and \ac{FER} considering two users, i.e., $N=2$. \ac{FER} is chosen as the performance metric instead of \ac{BER} because \ac{FER} is considered to be a more robust metric than \ac{BER} \cite{bocherer2015bandwidth}. Channel coding is performed using \ac{LDPC} codes specified in DVB-S2 \cite{Dvbs2x:15}. The block length used for all codes in DVB-S2 is fixed at 64800 bits. The optimal parameters such as  distribution, constellation spacing, and
\ac{FEC} rate are determined by \cref{algorithm:ch4} outlined in \cref{sec:ch4:RateAdaptation} with $\Rbf=0.05$. Without loss of generality, the received \ac{OSNR} is defined as $P_{r1}/\sigma$ in (dB) units.
 
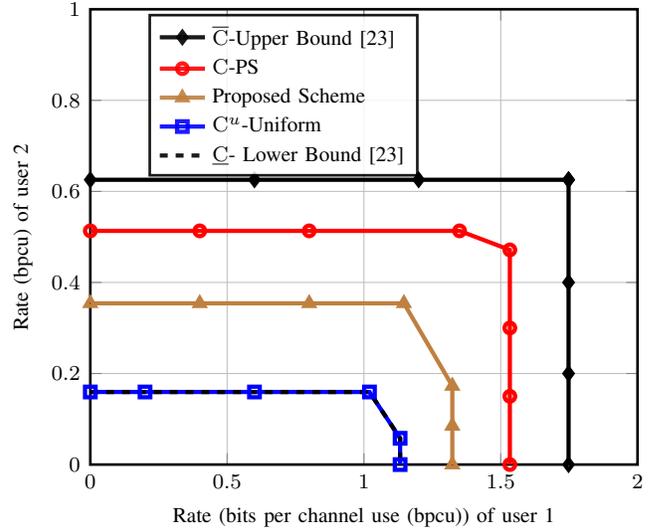
\begin{figure}[t!]
 	\centering
    \pgfplotsset{every axis/.append style={
		font=\footnotesize,
		line width=1pt,
		legend style={font=\footnotesize, at={(0.6,0.99)}},legend cell align=left},
} %
\pgfplotsset{compat=1.13}
	\begin{tikzpicture}
\begin{axis}[
xlabel near ticks,
ylabel near ticks,
grid=major,
xlabel={Rate (\ac{bpcu}) of user 1},
ylabel={Rate (\ac{bpcu}) of user 2},
yticklabel style={/pgf/number format},
width=1\linewidth,
legend entries={$\overline{\mathrm{C}}$-Upper Bound \cite{chaaban2017capacity}, $\mathrm{C}$-PS, Proposed Scheme, $\mathrm{C}^{u}$-Uniform, $\underline{\mathrm{C}}$- Lower Bound \cite{chaaban2017capacity}},
	xmin=0, xmax=2,
	ymin=0, ymax=1,
ylabel style={font=\footnotesize},
xlabel style={font=\footnotesize},
]
\addplot[black,solid,ultra thick,mark=diamond,mark size = 2pt,mark phase=1,mark repeat=1] table {Figures/CapacityRegion/Data_CapRegUB.dat};
\addplot[red,solid,ultra thick,mark=o,mark size = 2pt] table {Figures/CapacityRegion/Data_CapRegPS.dat};
\addplot[red,brown,ultra thick,mark=triangle,mark size = 2pt] table {Figures/CapacityRegion/Data_CapRegARnoma.dat};
\addplot[blue,solid,ultra thick,mark=square,mark size = 2pt] table {Figures/CapacityRegion/Data_CapRegUn.dat};
\addplot[black,dashed,ultra thick] table 
{Figures/CapacityRegion/Data_CapRegLB.dat};



\end{axis}
\end{tikzpicture}
 	\caption{Capacity region of \ac{NOMA} for $M=8$ and at \ac{OSNR} $P_{r1}/\sigma=6$ dB.} 
 	\label{fig:ch4:CapacityRegion}
\end{figure}

\begin{table}[b!]
  \caption{Performance gains of the $\scheme$ scheme.}
\centering
 \begin{tabular}{||c | c | c||}
 \hline
 Transmission Rate& User 1 Gap (dB) & User 2 Gap (dB) \\ [0.5ex] 
 \hline\hline
 0.75 & 1.3 & 2\\
 \hline
  1 & 0.5 & 1.5\\
 \hline
  2 & 0.5 & 0.9\\
 \hline
  2.5 & 0 & 0\\[1ex]
 \hline
\end{tabular}
\label{table:ch4:perf:gains}
\end{table}

The main performance benchmarks are as follows: 
\begin{itemize}
    \item Upper bound on the capacity of \ac{OMAC} channel denoted as $\overline{\mathrm{C}}_{j}$ and proposed in \cite{chaaban2017capacity}.
    \item Capacity of \ac{NOMA} denoted as $\CapNOMA$, which is the maximum of power-constrained $\Roimac$ for all feasible input distributions. The calculation of $\CapNOMA$ is shown in Appendix \ref{apendix:ch4:CapDerivation}.
    \item Uniform-distribution scheme, where output symbols of \ac{DM} and symbols resulted from the parity bits have uniform distribution. 
    \item \Ac{PCM} scheme proposed in \cite{he2019probabilistically} and extended to \ac{NOMA} case. In the \ac{PCM} scheme, the output symbols of \ac{DM} have an optimized pairwise distribution with equal probabilities for two consecutive constellation symbols while the symbols resulted from the parity bits remain uniformly distributed.
    \item \ac{GS} scheme, where the spacing values between constellation symbols are optimized while maintaining a uniform distribution.
\end{itemize}


\begin{figure}[t!]
 	\centering
	\begin{subfigure}{0.96\columnwidth}
 		\centering
		\pgfplotsset{every axis/.append style={
		font=\footnotesize,
		line width=1pt,
       legend style={font=\footnotesize, at={(0.85,1.3)}},legend columns=2},
} %
\pgfplotsset{compat=1.13}
	\begin{tikzpicture}
\begin{axis}[
xlabel near ticks,
ylabel near ticks,
grid=major,
xlabel={\ac{OSNR}$:P_{r1}/\sigma$ (dB)},
ylabel={Rate (\ac{bpcu})},
yticklabel style={/pgf/number format},
width=\linewidth,
legend entries={$\overline{\mathrm{C}}_{j}$-Upper Bound \cite{chaaban2017capacity},$\CapNOMA$-PS,$\mathrm{T}_{j}\left(\mathbf{p}^{\star}_{j}\right)$,$\mathrm{T}^{u}_{j}$},
	xmin=0, xmax=15,
	ymin=0, ymax=3.2,
ylabel style={font=\footnotesize},
xlabel style={font=\footnotesize},
]
\addplot[black,solid,very thick,mark=diamond,mark size = 2pt] table {Figures/CapArTxvsOSNR1/Data_CapUpB1.dat};
\addplot[red,solid,very thick,mark=o,mark size = 2pt] table {Figures/CapArTxvsOSNR1/Data_CapPS1.dat};
\addplot[ultra thick, red,dashed,mark=x,mark size = 2pt] table {Figures/CapArTxvsOSNR1/Data_TxPS1.dat};
\addplot+[const plot,no marks,ultra thick, blue, dashed,mark=square,mark size = 2pt] table {Figures/CapArTxvsOSNR1/Data_TxUn1.dat};
 \node[label={20:{$\Rfec^{\star}\!\!=\!\!\frac{1}{2}$}}] (a) at (0.02,0.13) {};
 \node[label={25:{$\frac{4}{5}$}}] (a) at (3,0.44) {};
\node[label={25:{$\frac{5}{6}$}}] (a) at (8,1.75) {};
\node[label={25:{$\frac{9}{10}$}}] (a) at (12.5,2.65) {};
\end{axis}
\end{tikzpicture}
 		\caption{User 1}
            \vspace{0.5cm}
 		\label{fig:ch4:ARUser1}
 	\end{subfigure}
        
 	\begin{subfigure}{0.96\columnwidth}
 		\centering
        \pgfplotsset{every axis/.append style={
		font=\footnotesize,
		line width=1pt,
		legend style={font=\footnotesize, at={(0.53,0.97)}},legend cell align=left},
} %
\pgfplotsset{compat=1.13}
	\begin{tikzpicture}
\begin{axis}[
xlabel near ticks,
ylabel near ticks,
grid=major,
xlabel={\ac{OSNR}$:P_{r1}/\sigma$ (dB)},
ylabel={Rate (\ac{bpcu})},
yticklabel style={/pgf/number format},
width=\linewidth,
	xmin=6, xmax=20,
	ymin=0, ymax=3.2,
ylabel style={font=\footnotesize},
xlabel style={font=\footnotesize},
]
\addplot[black,solid,very thick,mark=diamond,mark size = 2pt] table {Figures/CapArTxvsOSNR2/Data_CapUpB2.dat};
\addplot[red,solid,very thick,mark=o,mark size = 2pt] table {Figures/CapArTxvsOSNR2/Data_CapPS2.dat};
\addplot[ultra thick, red, dashed,mark=x,mark size = 2pt] table {Figures/CapArTxvsOSNR2/Data_TxPS2.dat};
\addplot+[const plot,no marks,ultra thick, blue, dashed] table {Figures/CapArTxvsOSNR2/Data_TxUn2.dat};
 \node[label={25:{$\Rfec^{\star}\!\!\!=\!\!\frac{1}{3}$}}] (a) at (6,0.1) {};
 \node[label={25:{$\frac{4}{5}$}}] (a) at (9,0.53) {};
\node[label={25:{$\frac{8}{9}$}}] (a) at (14.5,1.95) {};
\node[label={25:{$\frac{9}{10}$}}] (a) at (18.5,2.65) {};
\end{axis}
\end{tikzpicture}
 		\caption{User 2}
 		\label{fig:ch4:ARUser2}
 	\end{subfigure}
 	\caption{Transmission rates vs \ac{OSNR} for $M=8$.} 
 	\label{fig:ch4:AR_vs_SNR}
 \end{figure}
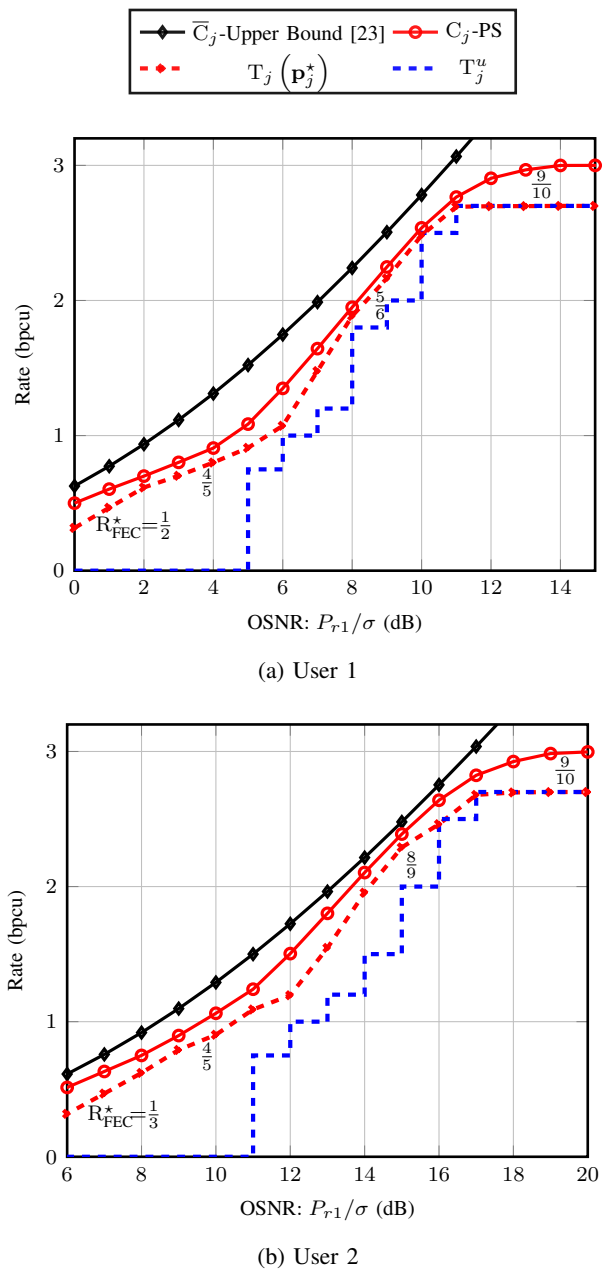%

In \cref{fig:ch4:CapacityRegion}, we present the capacity region characteristics in the \ac{NOMA} \ac{VLC} system for $M=8$ at \ac{OSNR} $P_{r1}/\sigma=6$ dB. The capacity region of \ac{NOMA}, denoted as $\mathrm{C}$, is separated from the theoretical upper bound on the capacity of the \ac{OMAC} channel, denoted as $\overline{\mathrm{C}}$, by a maximum gap of $0.2$ \ac{bpcu}. The maximum \ac{AR} region of the proposed scheme exhibits an improvement of approximately $0.2$ \ac{bpcu} over the capacity region of the uniform-distribution scheme, denoted as $\mathrm{C}^{u}$. The lower bound on the capacity region of the \ac{OMAC} channel, denoted as $\underline{\mathrm{C}}$, coincides with $\mathrm{C}^{u}$ since they both utilize a uniform distribution. As a result, $\mathrm{C}^{u}$ is regarded as a lower bound for the remainder of this section.

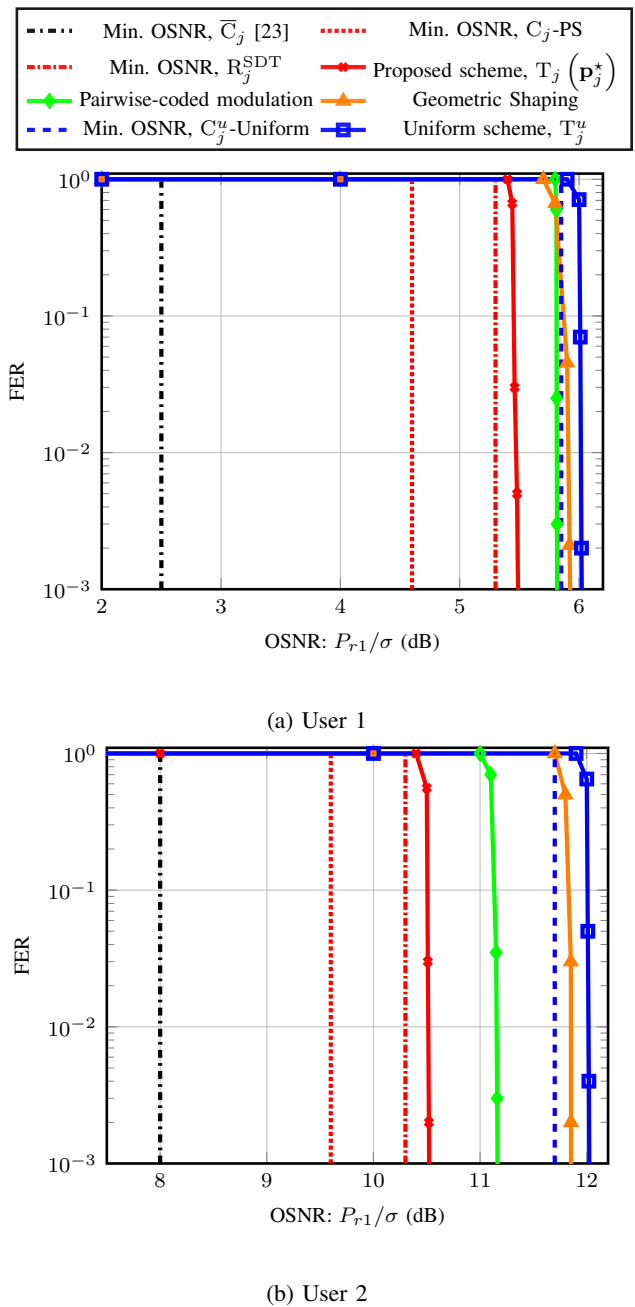
\begin{figure}[b!]
 	\centering
	\begin{subfigure}{0.95\columnwidth}
 		\centering
		  \pgfplotsset{every axis/.append style={
		font=\footnotesize,
		line width=1pt,
		legend style={font=\footnotesize, at={(1.06,1.4)}},legend columns=2},
} %
\pgfplotsset{compat=1.13}
	\begin{tikzpicture}
\begin{semilogyaxis}[
xlabel near ticks,
ylabel near ticks,
grid=major,
xlabel={\ac{OSNR}$:P_{r1}/\sigma$ (dB)},
ylabel={\ac{FER}},
width=0.98\linewidth, 
legend entries={Min. OSNR\mbox{,} $\overline{\mathrm{C}}_{j}$ \cite{chaaban2017capacity}, Min. OSNR\mbox{,} $\CapNOMA$-PS,Min. OSNR\mbox{,} $\mathrm{R}^{\mathrm{SDT}}_{j}$, Proposed scheme\mbox{,} $\mathrm{T}_{j}\left(\mathbf{p}^{\star}_{j}\right)$, \text{Pairwise-coded modulation},\text{Geometric Shaping},
Min. OSNR\mbox{,} $\mathrm{C}_{j}^{u}$-Uniform,Uniform scheme\mbox{,} $\mathrm{T}^{u}_{j}$},
	xmin=2, xmax=6.2,
	ymin=1e-3, ymax=1.1,
ylabel style={font=\footnotesize},
xlabel style={font=\footnotesize},
]
\addplot[black,dashdotted,ultra thick] coordinates {(2.5,1e-3) (2.5,1)};
\addplot[red,densely dotted,ultra thick] coordinates {(4.6,1e-3) (4.6,1)};
\addplot[red,densely dashdotted,ultra thick] coordinates {(5.3,1e-3) (5.3,0.2e-2)(5.3,0.5e-2) (5.3,1e-2) (5.3,0.2e-1) (5.3,1)};
\addplot[red,solid,ultra thick,mark=x,mark size = 2pt] table {Figures/FER1vsOSNR/Data_FER1.dat};
\addplot[green,solid,ultra thick,mark=diamond,mark size = 2pt] table {Figures/FER1vsOSNR/Data_FER1pw.dat};
\addplot[orange,solid,ultra thick,mark=triangle,mark size = 2pt] table {Figures/FER1vsOSNR/Data_FER1GS.dat};
\addplot[blue,dashed,ultra thick] coordinates {(5.85,1e-3) (5.85,1)};
\addplot[blue,solid,ultra thick,mark=square,mark size = 2pt] table {Figures/FER1vsOSNR/Data_FER1_u.dat};



\end{semilogyaxis}
\end{tikzpicture}
 		\caption{User 1}
            \label{fig:ch4:FERvsOSNRuser1}
 	\end{subfigure}
 	\begin{subfigure}{0.95\columnwidth}
 		\centering
        \pgfplotsset{every axis/.append style={
		font=\footnotesize,
		line width=1pt,
		legend style={font=\footnotesize, at={(0.75,0.95)}},legend cell align=left},
} %
\pgfplotsset{compat=1.13}
	\begin{tikzpicture}
\begin{semilogyaxis}[
xlabel near ticks,
ylabel near ticks,
grid=major,
xlabel={\ac{OSNR}$:P_{r1}/\sigma$ (dB)},
ylabel={\ac{FER}},
width=0.98\linewidth, 
	xmin=7.5, xmax=12.2,
	ymin=1e-3, ymax=1.1,
ylabel style={font=\footnotesize},
xlabel style={font=\footnotesize},
]
\addplot[black,dashdotted,ultra thick] coordinates {(8,1e-3) (8,1)};
\addplot[red,densely dotted,ultra thick] coordinates {(9.6,1e-3) (9.6,1)};
\addplot[red,densely dashdotted,ultra thick] coordinates {(10.3,1e-3) (10.3,0.2e-2)(10.3,0.5e-2) (10.3,1e-2) (10.3,0.2e-1) (10.3,1)};
\addplot[red,solid,ultra thick,mark=x,mark size = 2pt] table {Figures/FER2vsOSNR/Data_FER2.dat};
\addplot[green,solid,ultra thick,mark=diamond,mark size = 2pt] table {Figures/FER2vsOSNR/Data_FER2pw.dat};
\addplot[orange,solid,ultra thick,mark=triangle,mark size = 2pt] table {Figures/FER2vsOSNR/Data_FER2GS.dat};
\addplot[blue,dashed,ultra thick] coordinates {(11.7,1e-3) (11.7,1)};
\addplot[blue,solid,ultra thick,mark=square,mark size = 2pt] table {Figures/FER2vsOSNR/Data_FER2_u.dat};



\end{semilogyaxis}
\end{tikzpicture}
            \vspace{0.41cm}
 		\caption{User 2}
 		\label{fig:ch4:FERvsOSNRuser2}
 	\end{subfigure}
 	\caption{\ac{FER} of \ac{NOMA} for $M=8$ and rate $1$ \ac{bpcu}.} 
 	\label{fig:ch4:FERvsOSNR}
 \end{figure}%

\cref{fig:ch4:AR_vs_SNR} illustrates the capacity of \ac{NOMA}, $\CapNOMA$, its upper bound, $\overline{\mathrm{C}}_{j}$, \ac{TR} of the $\scheme$ scheme, $\mathrm{T}_{j}$, and \ac{TR} of the uniform-distribution scheme, $\mathrm{T}^{u}_{j}$, as a function of \ac{OSNR} for user $j$ and $M=8$. The optimal \ac{FEC} coding rate, $\Rfec^{\star}$, is adaptively adjusted based on the \ac{OSNR} for each user. The achieved performance gains of the $\scheme$ scheme, in terms of \ac{OSNR}, compared to the uniform-distribution scheme for different \acp{TR} of both users are summarized in \cref{table:ch4:perf:gains}. As depicted in \cref{fig:ch4:ARUser1}, $\mathrm{T}_{1}$ increases with the \ac{OSNR}. $\Rfec^{\star}$ also rises with \ac{OSNR} in order to maximize $\mathrm{T}_{1}$ while satisfying the rate constraint, i.e., $\mathrm{T}_{1}\leq\mathrm{R}^{\mathrm{SDT}}_{1}$, where $\RsdtoptOne$ is the \ac{AR} of user $1$. For example, $\Rfec^{\star}$ is $1/2$ at $1$ dB and increases to a maximum value of $9/10$ at $12$ dB. The optimum \ac{FEC} rate for the uniform-distribution scheme is also chosen to achieve maximum $\mathrm{T}^{u}_{1}$, but still ensures that it is below the capacity of the uniform signaling. As can be seen from \cref{fig:ch4:ARUser1}, at the \ac{OSNR} region $2-4$ dB, the $\scheme$ scheme operates within $0.05$ \ac{bpcu} and $0.25$ \ac{bpcu} from the capacity and its upper bound, respectively. Furthermore, at a rate of $0.75$ \ac{bpcu}, the $\scheme$ scheme outperforms uniform-distribution scheme by around $1.3$ dB. However, at high \ac{OSNR}, the gap between the \acp{TR} of the $\scheme$ and uniform-distribution schemes narrows for a fixed value of $M$. This is due to the fact that the conditional entropy of $X_{j} | \widehat{Y}_{j}$ tends to zero at high \ac{OSNR}, leaving only the entropy of $X_{j}$ in the \ac{AR} of user $j$. This entropy is maximized with a uniform distribution. Thus, it is recommended to switch to a higher $M$ for higher \ac{OSNR}.

\cref{fig:ch4:ARUser2} also shows the increase in $\Rfec^{\star}$ as a function of \ac{OSNR} for user $2$. The reason is similar as in the user $1$. At the \ac{OSNR} values of $6$ dB and $15$ dB, $\Rfec^{\star}$ is $1/3$ and $8/9$, respectively. Within the \ac{OSNR} range of $7-9$ dB, \cref{fig:ch4:ARUser2} demonstrates that the difference between the \ac{TR} of $\scheme$ scheme and its capacity, as well as its upper bound, is about $0.1$ \ac{bpcu} and $0.3$ \ac{bpcu}, respectively. Additionally, at a rate of $0.75$ \ac{bpcu}, the $\scheme$ scheme outperforms the uniform-distribution scheme by approximately $2$ dB. To achieve similar gains at higher \ac{OSNR}, it is recommended to switch to a higher value of $M$. It is worth noting that the accuracy of the capacity upper bound varies for different users, as evident from \cref{fig:ch4:AR_vs_SNR}. For example, the capacity upper bound provides a more accurate estimation for the user $2$ than for the user $1$.

To evaluate the performance of the $\scheme$ scheme, a comparison is made with the uniform-distribution, \ac{PCM}, and \ac{GS} schemes in terms of \ac{FER}. Monte Carlo simulations are conducted, and the results are presented in \cref{fig:ch4:FERvsOSNR}. The \acp{TR} of users $1$ and $2$ are fixed at $1$ \ac{bpcu}, and the decoding utilizes the \ac{MAP} values as computed in \eqref{eq:ch4:LLR}. The \ac{FER} of user $1$ is analyzed in \cref{fig:ch4:FERvsOSNRuser1}, and the minimum \acp{OSNR} required to achieve a \ac{TR} of $1$ \ac{bpcu} can be determined from \cref{fig:ch4:ARUser1}. The capacity upper bound, $\overline{\mathrm{C}}_{j}$, capacity of \ac{NOMA}, $\CapNOMA$, \ac{AR} of the $\scheme$ scheme, $\mathrm{R}^{\mathrm{SDT}}_{j}$, and capacity of \ac{NOMA} with uniform signaling, $\mathrm{C}_{j}^{u}$, provide reference points for these minimum \acp{OSNR}. Specifically, the respective values for user $1$ are $2.5$ dB, $4.6$ dB, $5.3$ dB, and $5.85$ dB. Furthermore, the $\scheme$ scheme employs the optimal distribution, constellation spacing, and \ac{FEC} rate, which are $\mathbf{p}^{\star}_{1}= \left[0.5297, 0, 0, 0, 0.1709, 0, 0.0168, 0.2825\right]$, $\Delta^{\star}_{1}=1.0042\times 10^{-5}$, and $\Rfec^{\star}=3/5$, respectively. For comparison, the \ac{PCM} scheme utilizes $\mathbf{p}^{pcm}_{1}=\left[0.2949, 0.2949, 0, 0, 0, 0, 0.2051,   0.2051\right]$, $\Delta^{pcm}_{1}=9.74\times 10^{-6}$, and $\Rfec^{pcm}=1/2$, while the \ac{GS} scheme optimizes the spacing values to $\Delta^{GS}_{1}=\left[0.009, 4.566, 1.400, 4.798, 0.005, 0.002, 0.002, 0.001\right]\times10^{-5}$. The uniform-distribution scheme uses $\mathbf{p}^{u}_{1}$, $\Delta^{u}_{1}= \left(2 P_{r1}\right)/(M+1)$, and $\Rfec^{u}=1/3$ as its parameters. It is worth noting that the \ac{FER} exhibits a phase transition phenomenon around the \ac{OSNR}, corresponding to the \ac{TR}. For example, at a \ac{FER} of $10^{-3}$, the $\scheme$ scheme outperforms the \ac{PCM}, \ac{GS}, and uniform-distribution schemes by approximately $0.3$ dB, $0.4$ dB, and $0.5$ dB, respectively. In \cref{fig:ch4:FERvsOSNRuser2}, we examine the minimum \acp{OSNR} required for user $2$ to achieve a \ac{TR} of $1$ \ac{bpcu}, which are derived from \cref{fig:ch4:ARUser2}. The respective minimum \ac{OSNR} values are $8$ dB, $9.6$ dB, $10.3$ dB, and $17.7$ dB for the capacity upper bound \cite{chaaban2017capacity}, capacity of \ac{NOMA}, \ac{AR} of the $\scheme$ scheme, and capacity of \ac{NOMA} with uniform signalling. Furthermore, we determine the optimal parameters for user $2$ in the $\scheme$ scheme, which are $\mathbf{p}^{\star}_{2}=\left[0.7357, 0, 0, 0, 0, 0, 0, 0.2643\right]$, $\Delta^{\star}_{2}=3.5\times 10^{-6}$, and $\Rfec^{\star}=3/4$. Similarly, the \ac{PCM} scheme employs $\mathbf{p}^{pcm}_{2}=\left[0.3659, 0.3659, 0, 0, 0, 0, 0.1341, 0.1341\right]$, $\Delta^{pcm}_{2}=3.63\times 10^{-6}$, and $\Rfec^{pcm}=1/2$, while the \ac{GS} scheme utilizes $\Delta^{GS}_{2}=\left[0, 0.1119, 0.0327, 0.1195, 0, 0, 0, 0\right]$ as the optimal spacing values for user $2$. We can observe that, at a \ac{FER} of $10^{-3}$, the $\scheme$ scheme exhibits a gain of $0.6$ dB, $1.3$ dB, and $1.5$ dB over the \ac{PCM}, \ac{GS}, and uniform-distribution schemes for user $2$, respectively.

\begin{figure}
 	\centering
	\begin{subfigure}{0.4\textwidth}
 		\centering
		\pgfplotsset{every axis/.append style={
		font=\normalsize,
		line width=1pt,
		legend style={font=\footnotesize, at={(0.98,0.8)}},legend cell align=left},
} %
\pgfplotsset{compat=1.13}
	\begin{tikzpicture}
\begin{axis}[
xlabel near ticks,
ylabel near ticks,
grid=major,
xlabel={$s_{1}$},
ylabel={$\mathbb{P}\{S_{1}=s_{1}\}$},
yticklabel style={/pgf/number format},
width=0.95\linewidth,
height=0.9\textwidth,
legend entries={\ac{OSNR}$=2$ dB,\ac{OSNR}$=8$ dB
	, \ac{OSNR}$=18$ dB},
	xmin=0, xmax=1.2e-4,
	ymin=0, ymax=1,
ylabel style={font=\large},
xlabel style={font=\large},
]
\addplot+[ycomb,red,mark=o, mark size = 3pt,solid,very thick] table {Figures/DistributionUser1/Data_user1Distr1.dat};
\addplot+[ycomb,blue,mark=x,mark size = 3pt,solid,very thick] table {Figures/DistributionUser1/Data_user1Distr2.dat};
\addplot+[ycomb,purple,mark=triangle,mark size = 3pt,solid,very thick] table {Figures/DistributionUser1/Data_user1Distr3.dat};



\end{axis}
\end{tikzpicture}
		\caption{User 1}
 		\label{fig:ch4:distr1}
 	\end{subfigure}
 	\begin{subfigure}{0.4\textwidth}
 		\centering
        \pgfplotsset{every axis/.append style={
		font=\normalsize,
		line width=1pt,
		legend style={font=\footnotesize, at={(0.98,0.8)}},legend cell align=left},
} %
\pgfplotsset{compat=1.13}
	\begin{tikzpicture}
\begin{axis}[
xlabel near ticks,
ylabel near ticks,
grid=major,
xlabel={$s_{2}$},
ylabel={$\mathbb{P}\{S_{2}=s_{2}\}$},
yticklabel style={/pgf/number format},
width=0.95\linewidth,
height=0.9\textwidth,
legend entries={\ac{OSNR}$=6$ dB,\ac{OSNR}$=14$ dB
	, \ac{OSNR}$=18$ dB},
	xmin=0, xmax=4.5e-5,
	ymin=0, ymax=1,
ylabel style={font=\large},
xlabel style={font=\large},
]
\addplot+[ycomb,red,mark=o, mark size = 3pt,solid,very thick] table {Figures/DistributionUser2/Data_user2Distr1.dat};
\addplot+[ycomb,blue,mark=x,mark size = 3pt,solid,very thick] table {Figures/DistributionUser2/Data_user2Distr2.dat};
\addplot+[ycomb,purple,mark=triangle,mark size = 3pt,solid,very thick] table {Figures/DistributionUser2/Data_user2Distr3.dat};



\end{axis}
\end{tikzpicture}
 		\caption{User 2}
 		\label{fig:ch4:distr2}
 	\end{subfigure}
        \caption{Optimized \ac{PMF} for the constellation symbols when $M=8$, for different instantaneous \acp{OSNR}.}
 	\label{fig:ch4:distr}
 \end{figure}
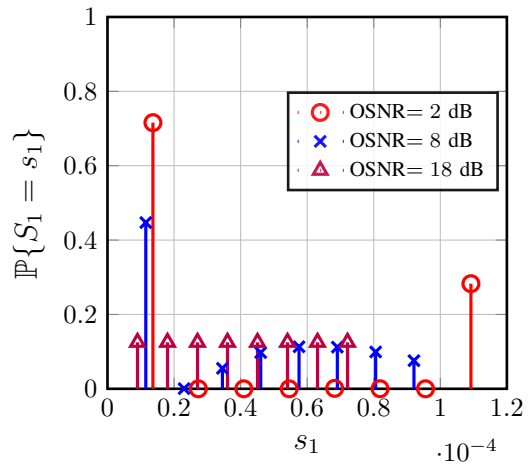
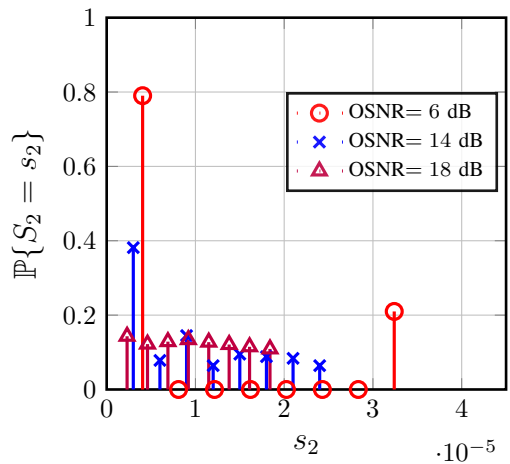%

\cref{fig:ch4:distr} provides a visualization of the optimized constellation symbols in the $\scheme$ scheme for each user at different \acp{OSNR}. At low \acp{OSNR}, as shown in \cref{fig:ch4:distr1} and \cref{fig:ch4:distr2}, symbols with the lowest amplitudes are assigned the highest probabilities. This assignment results in a larger inter-symbol spacing $\Delta_{j}^{\star}$, which satisfies the average optical power constraint and facilitates symbol detection in the presence of noise and interference. For example, the first and last constellation symbols of user $1$ have probabilities of approximately $\mathbb{P}\{S_{1}=\Delta_{1}^{\star}\}=0.7$ and $\mathbb{P}\{S_{1}=8\Delta_{1}^{\star}\}=0.3$ at the \ac{OSNR} of $2$ dB, while the probabilities of the remaining constellation symbols are negligible at the same \ac{OSNR}. However, as the \acp{OSNR} increase, the distribution of constellation symbols becomes nearly uniform, causing the inter-symbol spacing to reduce. This behavior is evident from \cref{fig:ch4:AR_vs_SNR}, where the gap between $\CapNOMA$ and $\CapNOMA^{u}$ diminishes at high \acp{OSNR}. Additionally, it is worth noting that the spacing between constellation points in \cref{fig:ch4:distr2} is smaller than in \cref{fig:ch4:distr1} due to the lower received power of user $2$ relative to user $1$.

\begin{figure}
 	\centering
    \pgfplotsset{every axis/.append style={
		font=\footnotesize,
		line width=1pt,
       legend style={font=\footnotesize, at={(0.99,0.3)}},legend cell align=left},
} %
\pgfplotsset{compat=1.13}
	\begin{tikzpicture}
\begin{axis}[
xlabel near ticks,
ylabel near ticks,
grid=major,
xlabel={\ac{OSNR}$:P_{r1}/\sigma$ (dB)},
ylabel={Sum rate (\ac{bpcu})},
yticklabel style={/pgf/number format},
width=\linewidth,
legend entries={Proposed for N=3, Uniform for N=3, Proposed for N=2, Uniform for N=2},
	xmin=0, xmax=20,
	ymin=0, ymax=10,
ylabel style={font=\footnotesize},
xlabel style={font=\footnotesize},
]
\addplot[black,solid,very thick,mark=x,mark size = 2pt] table {Figures/CapacityN3/Data_CapSumPSN3.dat};
\addplot[orange,solid,very thick,mark=diamond,mark size = 2pt] table {Figures/CapacityN3/Data_CapSumUnN3.dat};
\addplot[red,solid,very thick,mark=triangle,mark size = 2pt] table {Figures/CapacityN3/Data_CapSumPS.dat};
\addplot[blue,solid,very thick,mark=square,mark size = 2pt] table {Figures/CapacityN3/Data_CapSumUn.dat};
\end{axis}
\end{tikzpicture}
 	\caption{Maximum \ac{ASR} for $M=8$.} 
 	\label{fig:ch4:CapacityN3}
\end{figure}

In \cref{fig:ch4:CapacityN3}, we compare the maximum \ac{ASR} of our proposed and the uniform-based schemes with $N=2$ and $N=3$ users for a broad range of \acp{OSNR}. The results clearly demonstrate that our proposed scheme consistently outperforms the uniform-based scheme for $N=2$, showcasing the efficacy of our algorithm. With the introduction of an additional user, the performance gap in terms of the sum rate slightly widens due to increased number of users. For example, at $\text{\ac{OSNR}}=6$ dB, the performance gaps between the proposed and uniform-based schemes are around $0.8$ \ac{bpcu} and $1$ \ac{bpcu} for $N=2$ and $N=3$, respectively. To achieve similar gains at higher \ac{OSNR}, it is recommended to switch to a higher value of $M$. 

\section{Conclusion}\label{sec:ch4:Conclusion}
In this paper, we proposed a novel adaptive coded \ac{PS}-based \ac{NOMA} scheme to enhance the \ac{SE} in multiuser uplink \ac{VLC} communication scenario. We derived the \ac{AR} and \ac{TR} of the $\scheme$ scheme. We introduced an alternate optimization algorithm to determine the capacity-approaching distribution for each user for the given \ac{OSNR}. This algorithm adaptively adjusts the channel coding rate and constellation spacing based on the \ac{OSNR}. Our numerical results demonstrate that the $\scheme$ scheme achieves capacity of \ac{NOMA} with fine granularity. Overall, the $\scheme$ scheme provides significant improvements in \ac{SE} and \ac{FER} over the \ac{PCM}, \ac{GS} and uniform-distribution schemes, making it a promising candidate for practical implementation in future \ac{VLC} systems.


\begin{appendices}\label{apendix:ch4:CapDerivation}
\section{Capacity of NOMA}\label{apendix:ch4:CapDerivation}
To determine capacity of \ac{NOMA}, it is necessary to jointly optimize the distribution and spacing of the constellation symbols for every user. Consequently, there are $N$ separate optimization problems that need to be solved sequentially. 

First, after employing \ac{SIC} technique, only the Gaussian-distributed noise remains in the received signal of the last decoded user, i.e., $\widehat{Y}_{N}$. Hence, considering the power constraints described in \cref{sec:ch4:SignalModel}, we can determine the capacity of user $N$ by solving the following optimization problem:
\begin{subequations}\label{optimization:ch4:NCapDistrib}
\begin{alignat}{2}
\mathrm{C}_{N}=&\underset{\Delta_{N},\mathbf{p}_{N}}{\operatorname{maximize}} \quad &&\!\!\!\! \mathrm{R}_{N}\!\left(\Delta_{N},\mathbf{p}_{N},\mathbf{a}^{\star}_{\widehat{y}_{N}}\right)\\
&\text{subject to} \quad && \!\!\!\!\! \mathbf{s}^{T}_{N}\;\mathbf{p}_{N}\!=\!P_{rN}, \\
&\quad &&\!\!\!\!\! \frac{P_{rN}}{M}\leq\Delta_{N}\!\leq\!P_{rN},\\
& \quad &&\!\!\!\!\!  \sum_{i=1}^{M} p^{i}_{N}=1,  \\
& \quad &&\!\!\!\!\!   p^{i}_{N} \geq 0, \quad \forall i \in\mathbb{M}  \mbox{,}
\end{alignat}
\end{subequations}
where $\mathbf{a}^{\star}_{\widehat{y}_{N}}$ is an empty vector because there is no interference from other users in the received signal of the last decoded user, i.e., $\widehat{y}_{N}$, due to \ac{SIC}.

Considering the prior decoded user received signal $\widehat{Y}_{j<N}$, we observe the presence of interference from \ac{AWGN} and multiuser interference from users $j+1$ to $N$. To ensure optimal
distribution of the user $j<N$ signal, it is essential to account for this inter-user interference. For example, after obtaining $\mathbf{p}^{\star}_{N}$ and $\Delta^{\star}_{N}$ from \eqref{optimization:ch4:NCapDistrib}, the interference signal from user $N$ in $\widehat{Y}_{N-1}$ follows the distribution given by $\mathbf{p}^{\star}_{N}$ and has spacing value $\Delta^{\star}_{N}$. Thus, the capacity for user $j<N$ can be subsequently obtained by solving the following optimization problem:
\begin{subequations}\label{optimization:ch4:jCapDistrib}
\begin{alignat}{2}
\!\!\mathrm{C}_{j}=&\underset{\Delta_{j},\mathbf{p}_{j}}{\operatorname{maximize}} \quad && \!\!\!\mathrm{R}_{j}\!\left(\Delta_{j},\mathbf{p}_{j},\AIntrf^{\star}\right)\\
&\text{subject to} \quad && \!\!\!\!\! \mathbf{s}^{T}_{j}\;\mathbf{p}_{j}=P_{rj}, \\
&\quad &&\!\!\!\!\! \frac{P_{rj}}{M}\leq\Delta_{j}\!\leq\!P_{rj},\\
& \quad &&\!\!\!\!\!  \sum_{i=1}^{M} p^{i}_{j}=1,  \\
& \quad &&\!\!\!\!\!   p^{i}_{j} \geq 0, \quad \forall i \in\mathbb{M} \mbox{,}
\end{alignat}
\end{subequations}
where $\AIntrf^{\star}\triangleq\left[\DeltaIntf^{\star},\PdIntrf^{\star}\right]$. The optimal \acp{PMF} and constellation values that maximize capacity for each user are determined using the inverse decoding order of the \ac{SIC}, i.e., $N,N-1, \cdots, 1$.

\section{Deriving the KKT conditions}\label{apendix:ch4:KKT:conditions}
In the context of the convex optimization problem \eqref{optimization:ch4:MI:Tx:convex}, the Lagrangian function can be expressed as 
\begin{align}
    L&=-\mathrm{T}_{j}\left(\mathbf{p}_{j}\right)+\!\mu\!\!\left(\sum_{i=1}^{M} p^{i}_{j}-\!\!1\!\!\right)\!+\!\eta\!\left(\Rfec\sum_{i=1}^{M} \Delta\, i\, p^{i}_{j}\!-\!P_{j}\!\!\right)\nonumber \\
    &\phantom{=====}+\!\tau\left(\mathrm{T}_{j}\left(\mathbf{p}_{j}\right)-\Phip+\!\Rbf\right)\!\mbox{,}
\end{align}
where $P_{j}\triangleq P_{rj}-\left(1\!-\!\Rfec\right)\mathbf{s}^{T}_{j}\mathbf{p}^{u}_{j}$, and Lagrangian multipliers are denoted as $\mu$, $\eta$ and $\tau$. To derive the \ac{KKT} conditions, we differentiate the Lagrangian function $L$ with respect to $p^{i}_{j}$ and the Lagrangian multipliers $\mu$ and $\eta$, which are responsible for the equality constraints. Therefore, we have
\begin{align}
    \frac{dL}{dp^{i}_{j}}\!\!&=\!\!\Rfec\!\left(\log\!\left(p^{i}_{j}\right)\!+\!1\right)+\!\mu\!+\!\eta\, \Rfec \Delta\, i \nonumber \\
    &\phantom{========}+\!\tau\Rfec \Wiphat\!=\!0\mbox{,} 
    \label{eq:ch4:dL_dpi}
\end{align}
\begin{align}
    \frac{dL}{d\mu}=\sum_{i=1}^{M} p^{i}_{j}-1=0 \mbox{,}
    \label{eq:ch4:dL_dmu}
\end{align}
and
\begin{align}
    \frac{dL}{d\eta}=\Rfec\sum_{i=1}^{M} \Delta\, i\, p^{i}_{j}\!-\!P_{j}=0 \mbox{.}
    \label{eq:ch4:dL_eta}
\end{align}
The Lagrangian multiplier $\tau$, which is responsible for non-equality constraint, results in two conditions.\\
\textbf{Condition 1}: $\tau=0$\\
Equation \eqref{eq:ch4:dL_dpi} implies that
\begin{align}
    p^{i}_{j}=2^{-1-\frac{1}{\Rfec}\mu-\eta\,\Delta\, i}
    \mbox{.} 
    \label{eq:ch4:cond1_pi1}
\end{align}
By replacing $p^{i}_{j}$ in equation \eqref{eq:ch4:dL_dmu} with the expression from \eqref{eq:ch4:cond1_pi1}, $\mu$ becomes 
\begin{align}
    \mu=-\Rfec+\Rfec\log\left(\sum_{i^{\prime}=1}^{M} 2^{-\eta\,\Delta\, i^{\prime}}\right) \mbox{.}
    \label{eq:ch4:cond1_mu}
    \end{align}
Replacing $\mu$ in \eqref{eq:ch4:cond1_pi1} with \eqref{eq:ch4:cond1_mu}, $p^{i}_{j}$ is
\begin{align}
     p^{i}_{j}&=2^{-\eta\Delta i-\log\left(\sum\limits_{i^{\prime}=1}^{M} 2^{-\eta\,\Delta\, i^{\prime}}\right)}=\frac{2^{-\eta\,\Delta\, i}}{\sum\limits_{i^{\prime}=1}^{M} 2^{-\eta\,\Delta\, i^{\prime}}} \mbox{.} 
     \label{eq:ch4:cond1pif}
\end{align}
Due to the power constraint being an equality, the power constraint can be expressed as
\begin{align}
    \sum\limits_{i^{\prime}=1}^{M}\left(\Rfec\Delta\, i-\!P_{j}\right)2^{-\eta\Delta i}=0 \mbox{,}
    \label{eq:ch4:cond1_etaf}
\end{align}
where \eqref{eq:ch4:cond1_etaf} is derived by substituting $p^{i}_{j}$ in \eqref{eq:ch4:dL_eta} with \eqref{eq:ch4:cond1pif}.

Because $\tau=0$, the rate constraint should be satisfied with strict inequality, i.e., 
\begin{align}
&\mathrm{T}_{j}\left(\mathbf{p}_{j}\right)-\Phip+\!\Rbf<0 \nonumber \\
&\Rfec\!\sum\limits_{i=1}^{M}p^{i}_{j} \Wiphat\!+ \nonumber \\
&\phantom{======}+\!\left(\Rfec\!-\!1\right)\Ruopt\!\!+\!\Rbf<0 \mbox{.} 
\label{eq:ch4:cond1tau}
\end{align}
Substituting $p^{i}_{j}$ in \eqref{eq:ch4:cond1tau} with \eqref{eq:ch4:cond1pif} implies that
\begin{align}\label{eq:ch4:cond1tauf}
    \sum_{i=1}^{M} \left[\Rfec \Wiphat-D\right]2^{-\eta\,\Delta\, i}<0 \mbox{,} 
\end{align}
where $D=\left(1-\Rfec\right)\Ruopt-\Rbf$. \\

\textbf{Condition 2:} $\tau>0$ \\
Equation \eqref{eq:ch4:dL_dpi} implies that
\begin{align}
    p^{i}_{j}=2^{-1-\frac{1}{\Rfec}\mu-\eta\,\Delta\, i-\!\tau \Wiphat\!} \mbox{.}
    \label{eq:ch4:cond3_pi1}
\end{align}
Replacing $p^{i}_{j}$ in \eqref{eq:ch4:dL_dmu} with \eqref{eq:ch4:cond3_pi1}, $\mu$ becomes
\begin{align}
   \mu=\!-\Rfec\!+\!\Rfec\log\!\!\left(\sum_{i^{\prime}=1}^{M} 2^{-\eta\,\Delta\, i^{\prime}-\!\tau \Wiprimephat}\!\!\right)\!\!\mbox{.}
   \label{eq:ch4:cond3_mu}
\end{align}
Replacing $\mu$ in \eqref{eq:ch4:cond3_pi1} with \eqref{eq:ch4:cond3_mu}, $p^{i}_{j}$ is
\begin{align}
     p^{i}_{j}&=2^{-\eta\Delta i-\!\tau \Wiphat-\log\left(\sum\limits_{i^{\prime}=1}^{M} 2^{-\eta\,\Delta\, i^{\prime}-\!\tau \Wiprimephat}\right)}\nonumber \\
     &=\frac{2^{-\eta\,\Delta\, i-\!\tau \Wiphat}}{\sum\limits_{i^{\prime}=1}^{M} 2^{-\eta\,\Delta\, i^{\prime}-\!\tau \Wiprimephat}} \mbox{.}
     \label{eq:ch4:cond3pif}
\end{align}

Considering the power constraint as an equality, the following equation should be satisfied 
\begin{align}
    \sum\limits_{i^{\prime}=1}^{M}\left(\Rfec\Delta\, i-\!P_{j}\right)2^{-\eta\,\Delta\, i-\!\tau \Wiphat}=0 \mbox{,}
    \label{eq:ch4:cond3_etaf}
\end{align}
where \eqref{eq:ch4:cond3_etaf} is derived by substituting $p^{i}_{j}$ in \eqref{eq:ch4:dL_eta} with \eqref{eq:ch4:cond3pif}.

When $\tau>0$, the rate constraint is required to be satisfied with equality, i.e., 
\begin{align}
&\Rfec\!\sum\limits_{i=1}^{M}p^{i}_{j} \Wiphat\!+ \nonumber \\
&\phantom{======}+\!\left(\Rfec\!-\!1\right)\Ruopt\!\!+\!\Rbf=0 \mbox{.} 
\label{eq:ch4:cond3tau}
\end{align}
Substituting $p^{i}_{j}$ in \eqref{eq:ch4:cond3tau} with \eqref{eq:ch4:cond3pif} implies
\begin{align}\label{eq:ch4:cond3tauf}
    \!\!\!\sum_{i=1}^{M}\!\left[\Rfec \Wiphat\!-\!\!D\right]\frac{2^{-\eta\,\Delta\, i}}{2^{\!\tau \Wiphat}}\!=\!0 \mbox{.} 
\end{align}
\end{appendices}

\bibliographystyle{IEEEtran}
\bibliography{ref.bib}

\begin{IEEEbiography}[{\includegraphics[width=1in,height=1.25in,clip,keepaspectratio]{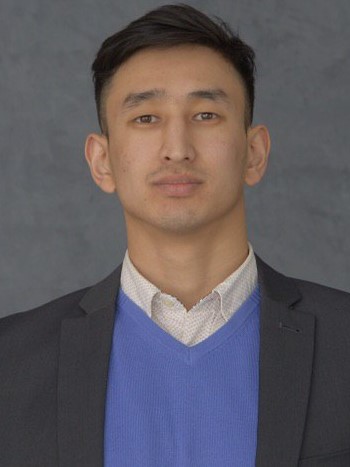}}]{Amanat Kafizov} received his B.Sc. and M.Sc. degrees in Electrical and Computer Engineering from Nazarbayev University, Kazakhstan, in 2019 and from King Abdullah University of Science and Technology (KAUST), Saudi Arabia, in 2021, respectively. Currently, he is a Ph.D. student at KAUST. His academic journey has been enriched by diverse experiences, including a summer research internship at the University at Buffalo, New York, USA, in 2018, where he developed OFDM system for Terahertz communication. Later, from 2023 to 2024, he worked as a research engineer at CERN, Geneva, Switzerland, focusing on the high granularity calorimeter (HGCAL). His research interests include visible light communication, reconfigurable intelligent surfaces, coded modulation, uneven distribution, and machine learning.
\end{IEEEbiography}

\begin{IEEEbiography}[{\includegraphics[width=1in,height=1.25in,clip,keepaspectratio]{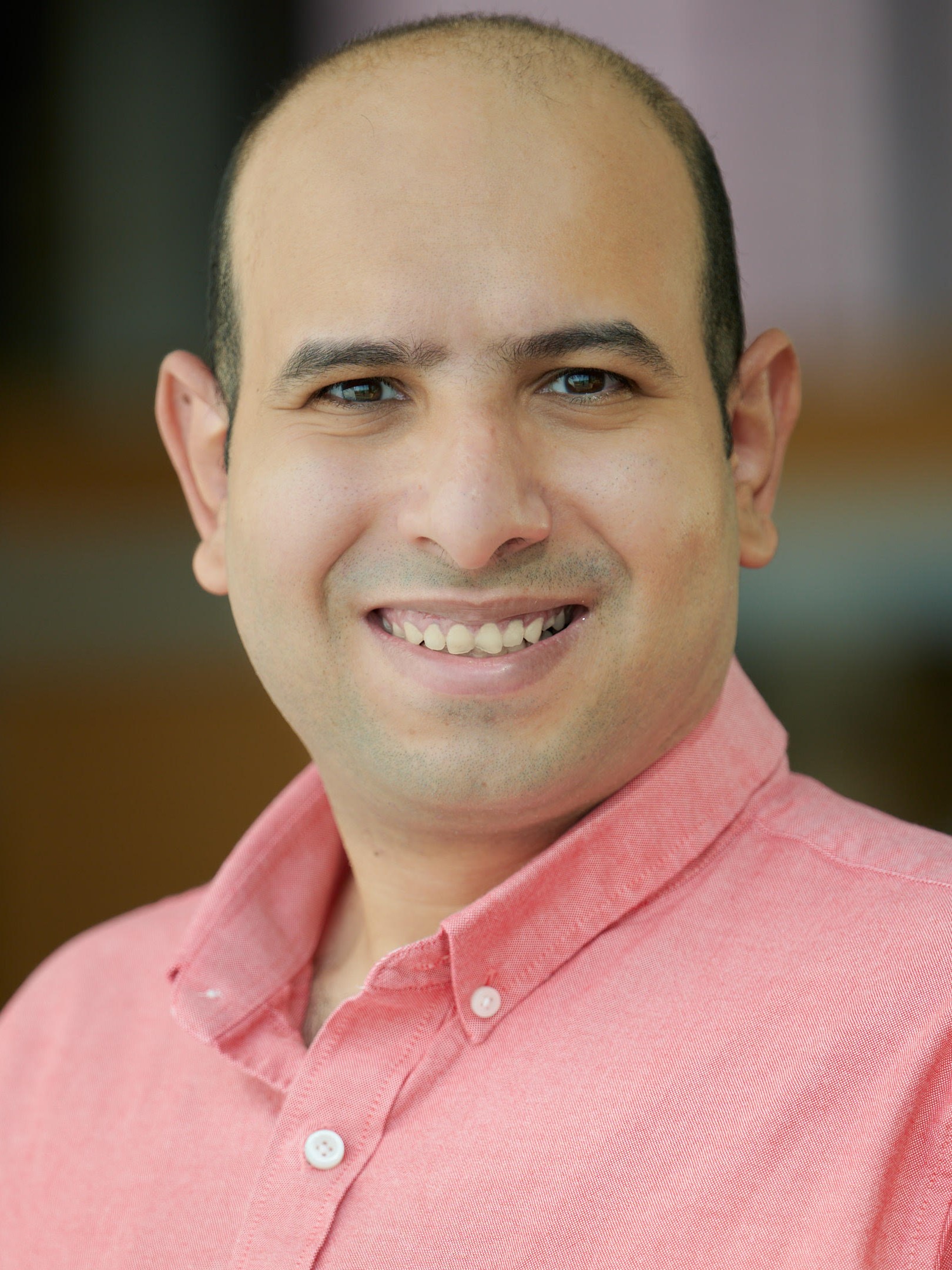}}]{Ahmed Elzanaty} (S12-M18-SM21) received the
Ph.D. degree (excellent cum laude) in electronics,
telecommunications, and information technology
from the University of Bologna, Italy, in 2018,
where he was a Research Fellow from 2017 to
2019. He was a Postdoctoral Fellow with
the King Abdullah University of Science and
Technology, Saudi Arabia. He is currently a Lecturer (Assistant Professor) at the Institute for Communication Systems, University of Surrey, United Kingdom. He has participated in several national and European projects, such as GRETA and EuroCPS. His research interests include the design and performance analysis
of wireless communications and localization systems, cellular network design
with EMF constraints, coded modulation, wireless localization, compressive
sensing, and distributed training of neural networks.
\end{IEEEbiography}

\begin{IEEEbiography}[{\includegraphics[width=1.1in,height=1.25in,clip,keepaspectratio]{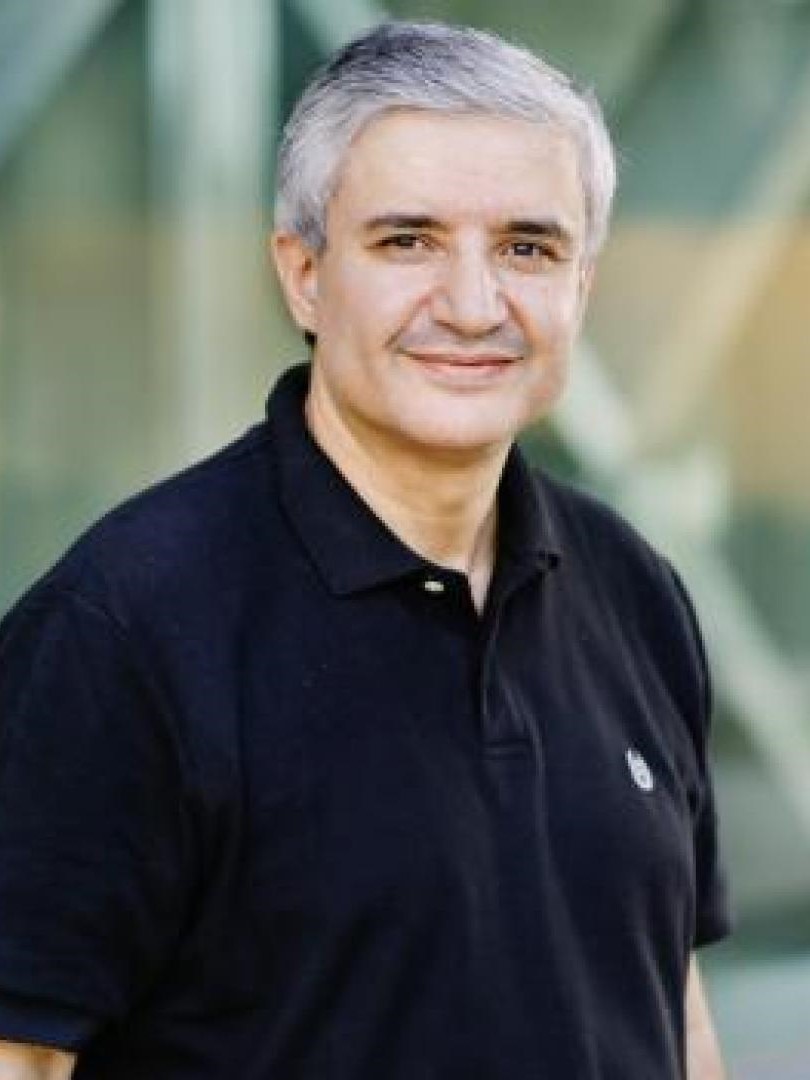}}]{Mohamed-Slim Alouini}
(S94-M98-SM03-F09) was born in Tunis, Tunisia. He received the Ph.D.
degree in Electrical Engineering from the California Institute of Technology (Caltech) in 1998. He served
as a faculty member at the University of Minnesota then in the Texas A\&M
University at Qatar before joining in 2009 the King Abdullah
University of Science
and Technology (KAUST) where he is now the Al-Khawarizmi Distinguished
Professor of Electrical and Computer Engineering. Prof. Alouini is a Fellow
of the IEEE and OPTICA (Formerly the Optical Society of America (OSA)). He
is currently particularly interested in addressing the technical challenges
associated with the uneven distribution, access to, and use of information
and communication technologies in rural, low-income, disaster, and/or
hard-to-reach areas.
\end{IEEEbiography}

\end{document}